\documentclass[11pt,english]{article}
\usepackage{color}
\usepackage{amsmath}
\usepackage{amssymb}
\usepackage[margin=1in]{geometry}
\usepackage{subfigure}
\usepackage{harvard}
\usepackage{lscape}
\usepackage{amsthm}
\usepackage{indentfirst}
\usepackage[bottom]{footmisc}
\usepackage{amsfonts}
\usepackage{rotating}
\usepackage[]{authblk}
\usepackage{setspace}

\newtheorem{theorem}{Theorem}
\newtheorem*{theorem*}{Theorem}
\newtheorem{corollary}{Corollary}
\newtheorem{definition}{Definition}
\newtheorem{definition*}{Definition}
\newtheorem*{setting*}{Setting}
\newtheorem{lemma}{Lemma}
\newtheorem{lemma*}{Lemma}
\newtheorem{proposition}{Proposition}

\newcommand{\msf}[1]{\mathsf{#1}}

\newcommand{\mbb}[1]{\mathbb{#1}}

\newcommand{\Exp}{\mathbb{E}}

\newcommand{\ord}{\lceil n\beta \rceil}

\newcommand{\trsp}{^\top}

\title{Performance-based regularization\\ in \\mean-CVaR portfolio optimization}
\author{Noureddine El Karoui\thanks{Department of Statistics, University of California, Berkeley.
Email: nkaroui@stat.berkeley.edu}~~~
Andrew E.B. Lim\thanks{Department of Industrial of Engineering \& Operations Research, University of California, Berkeley.
Email: lim@ieor.berkeley.edu}
~~~Gah-Yi Vahn\thanks{Corresponding author. Department of Industrial of Engineering \& Operations Research, University of California, Berkeley.
Email: gyvahn@ieor.berkeley.edu~ }}
\date{\today}

\begin{document}
\maketitle


\textbf{Abstract}
We introduce performance-based regularization (PBR), a new approach to addressing estimation risk in data-driven
optimization, to mean-CVaR portfolio optimization. We assume the available log-return data is iid, and detail the
approach for two cases: nonparametric and parametric (the log-return  distribution  belongs  in  the elliptical  family).
The nonparametric PBR method penalizes portfolios with large variability in mean and CVaR estimations.
The parametric PBR method solves the empirical  Markowitz problem instead of the empirical  mean-CVaR problem,  as the solutions  of  the Markowitz and mean-CVaR problems are equivalent when the log-return distribution is elliptical.
We derive the asymptotic behavior of the nonparametric PBR solution,
which leads to insight into the effect of penalization, and justification of the parametric PBR method.
We also show via simulations that the PBR methods produce efficient frontiers that are, on average, closer to the population
efficient frontier than the empirical approach to the mean-CVaR problem, with less variability.
\strut
\vspace{0.3cm}

\strut

\setstretch{1}

\newpage
\section{Introduction}
In recent years, there has been a growing interest in Conditional Value-at-Risk (CVaR) as a financial risk measure.
This interest is based on two key advantages of CVaR over Value-at-Risk (VaR), the risk measure of choice in the
financial industry over the last twenty years. Firstly, $CVaR(\beta)$, the conditional expectation
of losses in the top $100(1-\beta)\%$ ($\beta = 0.95, 0.99$ are typical values used in industry),
is more informative about the tail end of the loss distribution than $VaR(\beta)$, which is only the
\textit{threshold} for losses in the top $100(1-\beta)\%$. Secondly, CVaR satisfies the four coherence axioms of
\citeasnoun{Artzner-1998} [\citeasnoun{Acerbi-Tasche-2001a}], whereas VaR fails the subadditivity requirement.

Portfolio optimization with CVaR as a risk measure is first studied by \citeasnoun{Rockafellar-and-Uryasev-2000},
who show that empirical CVaR minimization can be formulated as a linear program. Subsequent works include CVaR
optimization for a portfolio of credit instruments [\citeasnoun{andersson2001credit}] and derivatives
[\citeasnoun{alexander2006derivatives}], and portfolio optimization based on extensions of CVaR [\citeasnoun{mansini2007}].
However, most discussions of CVaR in portfolio optimization to date are concerned with formulation and tractability of the
problem, and assume full knowledge of the distribution of the portfolio loss. In practice, one cannot
ignore the fact that the loss distribution is not known and must be estimated from historical data, constructed from
expert knowledge, or a combination of both. Naive estimation of the loss distribution can pose serious problems ---
\citeasnoun{Lim-Shant-Vahn-2010} demonstrates how fragile the solution to the empirical mean-CVaR problem is, even in
the ideal situation of having iid Gaussian log-return data.

The issue of estimation errors in portfolio optimization is not, however, new knowledge. The estimation issue for the
classical Markowitz (mean-variance) problem has been raised as early as 1980 [\citeasnoun{jobson-korkie}].
There have since been many suggestions for mitigating this issue for the Markowitz problem; two main
approaches are robust optimization [\citeasnoun{goldfarb-iyengar}] and what we call ``standard regularization''
[\citeasnoun{chopra1993}, \citeasnoun{frost-savarino}, \citeasnoun{jagannathan2003risk}, \citeasnoun{demiguel2009generalized}].
The robust optimization approach is to take the source of uncertainty (e.g. the asset log-returns, or its distribution),
specify an uncertainty set about the source, and minimize the worst-case return-risk problem over this uncertainty set.
The standard regularization approach is to solve the empirical mean-variance problem, but with a constraint on the size
of the solution, as measured by $L_2$ or a more generalized norm. The term ``regularization'' is adopted from statistics
and machine learning, where it refers to controlling for the size of the decision variable for better out-of-sample
performance [\citeasnoun{elem-stat-learn}]. Both robust optimization and standard regularization approaches have
been studied for the mean-CVaR problem; \citeasnoun{gotoh-shinozaki-robust} and \citeasnoun{zhu2009worst} show
implementations of the robust optimization approach when the source of uncertainty is, respectively, the log-return
vector and the log-return distribution, and \citeasnoun{gotoh-takeda} demonstrates implementation of standard regularization.

In this paper, we propose performance-based regularization (PBR), a new approach to addressing estimation
risk in data-driven optimization, and illustrate this method for the mean-CVaR portfolio optimization problem.
We demonstrate PBR for two situations: the investor has nonparametric or parametric (specifically, the elliptical family
of distributions describe the log-returns) information on the log-returns.

The nonparametric PBR method penalizes portfolios with large variability in mean and CVaR
estimations. Specifically, we penalize the sample variances of the mean and CVaR estimators. The resulting
problem is a combinatorial optimization problem, however we show that its convex relaxation, a quadratically-constrained
quadratic program, is tight. The problem can be interpreted as a chance-constrained program that picks portfolios for
which approximate probabilities of deviations of the mean and CVaR estimations from their true values are constrained.

The parametric PBR method solves the empirical Markowitz problem instead of the empirical mean-CVaR
problem if the underlying log-return distribution is in the elliptical family (which includes Gaussian and $t$ distributions).
This is based on the observation that CVaR of a portfolio is a weighted sum of the portfolio mean and the portfolio variance
if the log-return distribution is in the elliptical family, resulting in the equivalence of the population efficient
frontiers\footnote{By ``population'' we mean having a perfect market knowledge.} of the Markowitz and mean-CVaR problems.
As we are striving to reach the population frontier with greater stability, it makes intuitive sense to use the empirical
Markowitz solution in lieu of the empirical mean-CVaR solution for this model.

The PBR methods are anticipated to enhance the performance by yielding solutions that are, on average, closer to
achieving the original objective (minimize the \emph{true} CVaR subject to \emph{true} return equal to some level).
As such, the PBR approach is fundamentally different from robust optimization, in that robust
optimization deals with the source of uncertainty to minimize the worst-case performance, whereas PBR deals with the
performance uncertainty to increase the average performance. Comparing to the statistics/machine learning literature,
PBR for the nonparametric case can be seen as an extension of standard regularization,
in that nonparametric PBR also constrains the decision variable, however does so indirectly
through penalizing the variability of mean and CVaR estimations.

Details of the nonparametric PBR method can be found in Sec.~\ref{subsec:penalty} and the parametric PBR
method in Sec.~\ref{subsec:Mark}. In Sec.~\ref{sec:theory}, we provide theoretical results for the PBR methods
after deriving the Central Limit Theorem for the nonparametric PBR solution.
In Sec.~\ref{sec:numerical}, we evaluate the PBR methods against the straight-forward approach of solving the empirical mean-CVaR problem
for three different log-return models via simulation experiments. We find that on average, the sample efficient frontiers of the PBR solutions
are closer to the population efficient frontier than those of the straight-forward approach.

\section{Mean-CVaR portfolio optimization}
\textbf{Notations.} Throughout the paper, we denote convergence in probability by $\overset{P}{\rightarrow}$ and in distribution by $\Rightarrow$.
The notation $X \overset{d}{=}Y$ for two random variables $X$ and $Y$ means they have the same distribution, and the symbol $X\sim \mathcal{D}$ is used to indicate that
the random variable $X$ follows some standard distribution $\mathcal{D}$.

\subsection{Setup}\label{subsec:demo}
An investor is to choose a portfolio $w\in\mathbb{R}^p$ on $p$ different assets. Her wealth is normalized to 1, so $w\trsp {1}_p=1$, where $1_p$ denotes $p\times 1$ vector of ones.
The log-returns of the $p$ assets is denoted by $X$, a $p\times 1$ random vector,
which follows some absolutely continuous distribution $F$ with twice continuously differentiable pdf
and finite mean $\mu$ and covariance $\Sigma$.
The investor wants to pick a portfolio that minimizes the CVaR of the portfolio loss at level $100(1-\beta)\%$, for some $\beta\in(0.5,1)$, while reaching an expected return $R$.
That is, she wants to solve the following problem:
\begin{equation}\tag{CVaR-pop}\label{eq:population problem}
\begin{array}{clll}
w_0=\underset{w}{\mathrm{argmin}} & CVaR(-w\trsp {X};\beta) && \\
s.t. &  w\trsp {\mu} = R &  \\
& w\trsp 1_p   =1,&
\end{array}
\end{equation}
where
\begin{equation}\label{def:CVaR RU}
CVaR(-w\trsp {X};\beta) := \underset{\alpha}{\min}~~\alpha+\frac{1}{1-\beta}\mbb{E}(-w\trsp {X}- \alpha)^+,
\end{equation}
as in \citeasnoun{Rockafellar-and-Uryasev-2000}.

In reality, the investor does not know the distribution $F$. We assume the investor observes $n$ iid realizations of asset returns, $\mathsf{X} = [X_1,\ldots,X_n] \in \mathbb{R}^{p\times n} $.
Then the most straight-forward thing is to solve the following problem, where plugged-in estimators replace the true CVaR and return values:
\begin{equation}\tag{CVaR-emp}\label{eq:empricial problem}
\begin{array}{clll}
\hat{w}_n = \underset{w}{\mathrm{argmin}} & \widehat{CVaR}_n(-w\trsp \mathsf{X};\beta) && \\
s.t. &  w\trsp \hat{\mu}_n = R &  \\
& w\trsp 1_p   =1&
\end{array}
\end{equation}
where
\begin{equation}\label{eq:CVAR est}
\widehat{CVaR}_n(-w\trsp \mathsf{X};\beta):= \underset{\alpha\in\mathbb{R}}{\min}~~
\alpha+\displaystyle \frac{1}{n(1-\beta)}\sum_{i=1}^{n}(-{w}\trsp X_{i}-\alpha)^{+},
\end{equation}
is a sample average estimator for $CVaR(-w\trsp {X};\beta)$ and $\hat{\mu}_n=n^{-1}\sum_{i=1}^nX_i$ is the sample mean of the
observed asset log-returns.

\subsection{Estimation risk of the empirical solution}\label{subsec:demo}
Asymptotically, as the number of observations $n$ goes to infinity (with $p$ constant), $\hat{w}_n$ converges in probability to $w_{0}$ [see Sec.~\ref{subsec:consistency} for details]. In practice, however, the investor has
a limited number of relevant observations. If, for example, there are $n=250$ iid daily observations, and the investor wishes to control the top $5\%$
of the losses, then there are only $250\times 0.05 = 12.5$ points to estimate the portfolio CVaR at level $\beta = 0.95$. For stock log-returns,
$n=250$ iid daily observations is rather optimistic; there is ample empirical evidence that suggests daily log-returns are
non-stationary over this period of time [\citeasnoun{mcneil2005QRM}]. Even for time scales with more evidence for stationarity (e.g. bi-weekly/montly), the stationarity tends to last for no more than 5 years [\citeasnoun{mcneil2005QRM}].

As a result, solving (\ref{eq:empricial problem}) using real data results in highly unreliable solutions. Let us illustrate this point, assuming an ideal market scenario.
There are $p=10$ stocks, with daily returns following a Gaussian distribution\footnote{the parameters are the sample mean and covariance matrix of data from 500 daily returns of 10 different US stocks from Jan 2009-- Jan 2011}:
$X\sim\mathcal{N}(\mu_{sim},\Sigma_{sim})$, and the investor has $n=250$ iid observations of $X$. In the following, we conduct an experiment similar to those found in \citeasnoun{Lim-Shant-Vahn-2010}, to evaluate the performance and reliability of solving (\ref{eq:empricial problem}) under this ideal scenario. Briefly, the experimental procedure is as follows:
\begin{itemize}
  \item Simulate 250 historical observations from $\mathcal{N}(\mu_{sim},\Sigma_{sim})$.
  \item Solve (\ref{eq:empricial problem}) with $\beta=0.95$ and some return level $R$ to find an instance of $\hat{w}_n$.
  \item Plot the realized return $\hat{w}_n\trsp \mu$ versus realized risk ${CVaR}(-\hat{w}_n\trsp X;\beta)$; this corresponds to one grey point in Fig. (\ref{Fig:empirical scatter}).
  \item Repeat for different values of $R$ to obtain a sample efficient frontier.
  \item Repeat many times to get a distribution of the sample efficient frontier.
\end{itemize}

The result of the experiment is summarized in Fig.~(\ref{Fig:empirical scatter}). The green curve corresponds to the population efficient frontier.
Each of the grey dots corresponds to a solution instance of (\ref{eq:empricial problem}). There are two noteworthy observations:
the solutions $\hat{w}_n$ are sub-optimal, and they are highly variable.
For instance, for a daily return of $0.1\%$, the CVaR ranges from $1.3\%$ to $4\%$.

\begin{figure}[h]
\centering
\includegraphics[width=0.6\paperwidth]{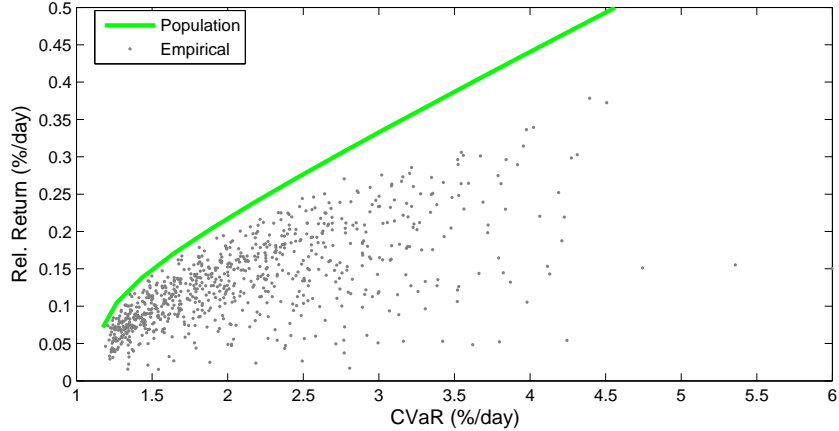}
\caption{\small Distribution of realized daily return (\%) vs. daily risk (\%) of empirical solution $\hat{w}$.
Green line represent the population frontier, i.e.~the efficient frontier corresponding to solving (\ref{eq:population problem}).}
\label{Fig:empirical scatter}
\end{figure}

In the following section, we introduce \emph{performance-based regularization} (PBR) as an approach to improve upon (\ref{eq:empricial problem}).
The PBR approach is so-called because its goal is to improve upon $\hat{w}_n$ in terms of its performance, i.e.
closeness to the population efficient frontier, ideally with less variability.
We describe PBR for two cases: the investor has nonparametric or parametric knowledge of the market.

\section{Performance-based regularization}\label{sec:penalized problem}
\subsection{Nonparametric case}\label{subsec:penalty}
In the nonparametric case, we assume the asset log-returns $X$ follows some distribution $P$ with finite mean $\mu$ and covariance $\Sigma$, and the investor has $n$ iid observations: $\mathsf{X} = [X_1,\ldots,X_n]\in\mbb{R}^{p\times n}$.
The nonparametric PBR approach to (\ref{eq:population problem}) is to solve the following problem:
\begin{eqnarray}\label{eq:penalized problem}
     \begin{array}{clll}
      \underset{w}{\min} & \widehat{CVaR}_n(-w\trsp \mathsf{X};\beta) &&\\
      s.t. &  w\trsp \hat{\mu}_n = R& \\
         & w\trsp 1_p  = 1& \\
         & P_1(w)\leq U_1& \\
         & P_2(w) \leq U_2&
\end{array}
\end{eqnarray}
where $P_1$ and $P_2$ are penalty functions that characterize the uncertainty associated with
$w\trsp \hat{\mu}_n$ and $\widehat{CVaR}_n(-w\trsp \mathsf{X};\beta)$ respectively.
The idea is to penalize decisions $w$ for which the uncertainty about the true values $w\trsp \mu$ and $CVaR(-w\trsp {X};\beta)$ is large.

What, then, are appropriate penalty functions? Recall that we are trying to find solutions that yield efficient frontiers that are
closer to the population efficient frontier, ideally with smaller variability.
Thus the variances of $w\trsp \hat{\mu}_n$ and $\widehat{CVaR}_n(-w\trsp \mathsf{X};\beta)$ make appropriate penalty functions,
as they characterize the deviation from the respective population values. The variance of $w\trsp \hat{\mu}_n$ is given by
$$Var(w\trsp \hat{\mu}_n) = \frac{1}{n^2}\sum_{i=1}^n Var(w\trsp X_i)=\frac{1}{n} w\trsp \Sigma w,$$
and the variance of $\widehat{CVaR}_n(-w\trsp \mathsf{X};\beta)$ is approximately equal to $\gamma_0^2/n(1-\beta)^2
=Var[\max(-w\trsp X - \alpha_\beta)]$, where
$$\alpha_\beta = \inf \{\alpha: P(-w\trsp X\geq \alpha) \leq 1-\beta\}$$
is the Value-at-Risk (VaR) of the portfolio $w$ at level $\beta$, due to the following lemma.
\begin{lemma}\label{lem:CV ast}
Suppose $\mathsf{X}=[X_1,\ldots,X_n]\overset{iid}{\sim} F$, where $F$ is absolutely continuous with twice continuously differentiable pdf.
Then
\begin{equation}
\frac{\sqrt{n}(1-\beta)}{\gamma_0}(\widehat{CVaR}_n(-w\trsp \mathsf{X};\beta)-CVaR(-w\trsp X;\beta)) \Rightarrow \mathcal{N}(0,1).
\end{equation}
\end{lemma}
\begin{proof}
See Appendix \ref{app:note on cvar}.
\end{proof}

Of course, we do not know the true variances, so we contend with
\emph{sample} variances of the estimators $w\trsp \hat{\mu}_n$ and $\widehat{CVaR}_n(-w\trsp \mathsf{X};\beta)$. That is, we consider
the following penalty functions:
\begin{eqnarray*}
P_1(w) &=& \displaystyle\frac{1}{n}w\trsp \hat\Sigma_nw,~\text{where}~\hat\Sigma_n = Cov(\mathsf{X}),\\
P_2(w) &=& \displaystyle\frac{1}{n(1-\beta)^2}{z}\trsp \Omega_n{z},~\text{where} \\
\Omega_n &=& \frac{1}{n-1}[I_n-n^{-1}1_n1_n\trsp],~I_n = n\times n~ \text{identity matrix},~\text{and} \\
z_i &=& \max(0, -w\trsp X_i-\alpha)~\text{for}~ i=1,\ldots, n.
\end{eqnarray*}

For the rest of this paper, we investigate the nonparametric PBR method with sample variance penalty functions.
Of course, this is just one particular choice, and it opens up the question of how different penalty functions affect the solution performance,
and whether there are such things as ``optimal'' penalty functions. These are difficult questions worthy of further research,
and we do not investigate them in this paper. Nevertheless,
we derive the asymptotic behavior of the solution of nonparametric PBR method in Sec.~\ref{sec:theory}, which gives us some insight into
i) how one could compare the effects of different penalty functions and ii) the first-order effect of many typical penalty functions.

The nonparametric PBR method with sample variance of return and CVaR estimators as penalties is:
\begin{align}
\begin{array}{rrl}
(\hat{\alpha}^v_n,\hat{w}^v_n,\hat{z}^v_{n})~~=~~\underset{\alpha,w,{z}}{\textrm{argmin}}&& \alpha+\displaystyle\frac{1}{n(1-\beta)}\sum_{i=1}^nz_i \\
\end{array}\notag\\
\begin{array}{crll}
      &s.t.~~~w\trsp \hat{\mu}_n &=& R  \\
      &    w\trsp 1_p  &=& 1  \\
      &    \displaystyle\frac{1}{n}w\trsp \hat\Sigma_nw  &\leq& U_1 \\
      &    \displaystyle\frac{1}{n(1-\beta)^2}{z}\trsp \Omega_n{z} &\leq&  U_2 \\
      &    z_i =  \max(0, -w\trsp X_i-\alpha),& i=&1,\ldots, n.  \\
\end{array}\tag{CVaR-pen}\label{eq:penalized problem-asympvar}
\end{align}

At first glance, (\ref{eq:penalized problem-asympvar}) is a combinatorial optimization problem due to the cutoff variables $z_i$, $i=1,\ldots,n$.
However, it turns out that the convex relaxation of (\ref{eq:penalized problem-asympvar}), a quadratically-constrained quadratic program (QCQP), is
tight, thus we can solve (\ref{eq:penalized problem-asympvar}) efficiently. Before stating the result, let us first introduce the convex relaxation of (\ref{eq:penalized problem-asympvar}):

\begin{align}
\underset{\alpha,w,{z}}{\min}&~~~~~~~ \alpha+\displaystyle\frac{1}{n(1-\beta)}\sum_{i=1}^nz_i \notag\\
&\begin{array}{crll}
 &s.t.~~~  w\trsp \hat{\mu}_n &= R  &(\nu_1) \\
         & w\trsp 1_p  &= 1 &(\nu_2)\\
         & \displaystyle\frac{1}{n}w\trsp \hat{\Sigma}_nw &\leq U_1 &(\lambda_1)\\
         & \displaystyle\frac{1}{n(1-\beta)^2}{z}\trsp \Omega_n{z} &\leq U_2 &(\lambda_2)\\
         & z_i  \geq 0& i=1,\ldots, n &({\eta}_1)\\
         & z_i  \geq -w\trsp X_i-\alpha, &i=1,\ldots, n &({\eta}_2)
\end{array}\tag{CVaR-relax}\label{eq:penalized problem-QCQP}
\end{align}
and its dual (where the dual variables correspond to the primal constraints as indicated above):
\begin{equation}\tag{CVaR-relax-d}\label{eq:dual of reg}
\begin{array}{clll}
      \underset{\nu_1,\nu_2,\lambda_1,\lambda_2,{\eta}_1,{\eta}_2}{\max}
      & g(\nu_1,\nu_2,{\eta}_1,{\eta}_2,\lambda_1,\lambda_2)&&\\
      s.t. & {\eta}_2\trsp 1_n=1  &&\\
         & \lambda_1\geq0,\lambda_2\geq0 &&\\
         & {\eta}_1\geq{0},{\eta}_2\geq{0} &&
\end{array}
\end{equation}
where
\begin{eqnarray*}
g(\nu_1,\nu_2,\lambda_1,\lambda_2,{\eta}_1,{\eta}_2) &=&  -\frac{n}{2\lambda_1}(\nu_1\hat{\mu}_n+\nu_21_p-\mathsf{X}{\eta}_2)\trsp \hat{\Sigma}_n^{-1}(\nu_1\hat{\mu}_n+\nu_21_p-\mathsf{X}{\eta}_2)\\
&& - \frac{n(1-\beta)^2}{2\lambda_2}({\eta}_1+{\eta}_2)\trsp \Omega_n^\dagger({\eta}_1+{\eta}_2) + R\nu_1+\nu_2-U_1\lambda_1-U_2\lambda_2,
\end{eqnarray*}
and $\Omega_n^\dagger$ is the Moore-Penrose pseudo inverse of the singular matrix $\Omega_n$.

We now show (\ref{eq:penalized problem-asympvar}) can be solved efficiently by its convex relaxation:
\begin{theorem}\label{theorem:QCQP}
Let $(\alpha^*,w^*,{z}^*,\lambda_1^*,\lambda_2^*,{\eta}^*_1,{\eta}^*_2)$ be the primal-dual optimal point of (\ref{eq:penalized problem-QCQP}) and (\ref{eq:dual of reg}).
If ${\eta}^*_2\neq 1_n/n$, then $(\alpha^*,w^*,{z}^*)$ is an optimal point of (\ref{eq:penalized problem-asympvar}).
Otherwise, if ${\eta}^*_2= 1_n/n$, we can find the optimal solution to (\ref{eq:penalized problem-QCQP}) by solving (\ref{eq:dual of reg}) with an additional constraint ${\eta}_1\trsp 1_n\geq\delta$, where $\delta$ is a constant $0<\delta\ll1$.
\end{theorem}

\begin{proof}
See Appendix \ref{app: qcqp proof}.
\end{proof}

\textbf{Remark 1 -- Bias introduced by penalty functions.}
\\Note that if the penalties induce active constraints (i.e. $U_1,U_2$ are small enough),
$\hat{w}_n^v$ does not converge to $w_0$ as $n\rightarrow\infty$, i.e. the penalty constraints introduce bias.
This is not a problem, however, because we are concerned with finite sample
performance, not asymptotic consistency. In Sec.~\ref{sec:numerical}, we see that
the bias introduced by the penalized solution is actually in the direction that improves performance in the return-risk space.

\textbf{Remark 2 -- Interpretation as chance-programming.}
\\Both $w\trsp \hat{\mu}_n$ and $\widehat{CVaR}_n(-w\trsp \mathsf{X};\beta)$ are asymptotically normally distributed, so
constraining their variances results in the reduction of the corresponding confidence intervals at some fixed level $\epsilon$.
Hence penalizing their variances can be interpreted as chance-programming [\citeasnoun{chance1958}]. Analytically, the chance constraint
on $|w\trsp \hat{\mu}_n-w\trsp \mu|$ can be transformed to a penalty constraint in the following manner:
\begin{eqnarray*}
P\left(|w\trsp \hat{\mu}_n-w\trsp \mu|\leq t \right) &\geq& 1-\epsilon  \\
\approx 2\Phi\left(\frac{t}{\sqrt{w\trsp \Sigma w/n}}\right)-1 &\geq& 1-\epsilon~~ \text{for large}~ n \\
\iff~~~~~~ \frac{1}{n}w\trsp \Sigma w &\leq& \left(\frac{t}{\Phi^{-1}(1-\epsilon/2)}\right)^2.
\end{eqnarray*}
\noindent That is, for a fixed level $\epsilon$, there is a one-to-one mapping between the parameter $U_1$ of the penalty constraint $w\trsp \hat{\Sigma}_n w/n\leq U_1$
and the parameter $t$ of the chance constraint. The (asymptotic) variance penalty on $\widehat{CVaR}_n(-w\trsp \mathsf{X};\beta)$
has a similar interpretation as a chance constraint.

However, the penalty method can be interpreted as chance programming only if we choose the variance of the respective estimators as the penalty functions.
Although we focus on the sample variance penalty function in this paper,
we assert that the penalty method need not be restricted to this particular choice.

\subsection{Parametric case}\label{subsec:Mark}
In the parametric case, we assume the asset log-returns follow an elliptical distribution; i.e. the level sets of the
distribution density function form ellipsoids. An elliptical distribution has a stochastic representation as follows
[see \citeasnoun{anderson1958introduction} or \citeasnoun{muirhead1982aspects}]:
\begin{equation}\label{eq:elliptical model}
X \overset{d}{=} \mu + Y\Sigma^{1/2}U
\end{equation}
where $\mu$ is the mean vector, $U$ is a $p\times 1$ random vector uniformly distributed on the $p$-dimensional sphere of
radius 1 (i.e. $U \overset{d}{=} Z_p/||Z_p||_2$, $Z_p\sim\mathcal{N}(0,I_p)$), and $Y$ is a non-negative random variable
independent of $U$. A special case is the Gaussian model:
choosing $Y = \chi_p$, we get $X\sim\mathcal{N}(\mu,\Sigma)$. The elliptical family of distributions can thus be thought
of as a generalization of the Gaussian family, and may be more reasonable for financial modeling because the non-random mixing
of covariances can capture non-trivial tail dependence and heavier tails [\citeasnoun{mcneil2005QRM}]. In particular,
$t$-distributions also belong in the elliptical family.

The parametric PBR method is to solve the empirical Markowitz problem instead of (\ref{eq:empricial problem})
if $X$ belongs in the elliptical family:
\begin{equation}\tag{Mark-emp}\label{eq:Markowitz}
\begin{array}{crll}
\hat{w}^M_{n}=\underset{w}{\text{argmin}} & w\trsp \hat{\Sigma}_nw && \\
s.t. &  w\trsp \hat{\mu}_n = &R &  \\
& w\trsp 1_p   =&1.&
\end{array}
\end{equation}
The method is based on Lemma~\ref{lem:equiv}, which shows that the
solutions of (\ref{eq:population problem}) and the population Markowitz problem [which is the same as (\ref{eq:Markowitz})
except with ($\Sigma,\mu$) replacing ($\hat{\Sigma}_n,\hat{\mu}_n$)] are equivalent if $X$ is elliptically distributed. Lemma~\ref{lem:equiv}
is an extension of results mentioned elsewhere [\citeasnoun{Rockafellar-and-Uryasev-2000}, \citeasnoun{DeGiorgi-2002}] that show the equivalence
of the solutions of (\ref{eq:population problem}) and the population Markowitz problem when $X$ is Gaussian.
However, to our knowledge, the implication that we can solve (\ref{eq:Markowitz}) in lieu of (\ref{eq:empricial problem})
to obtain a better-performing solution has not been asserted.

\begin{lemma}\label{lem:equiv}
Suppose $X\sim Ellip(\mu,\Sigma,Y)$ as in (\ref{eq:elliptical model}) and $Y>0$. Then the solution of the
population mean-CVaR problem (\ref{eq:population problem}) and the population Markowitz problem are equivalent.
\end{lemma}

\begin{proof}
The proof is straightforward: we show $CVaR(-w\trsp X;\beta)$ is a weighted sum of the portfolio mean $w\trsp \mu$ and portfolio
std $\sqrt{w\trsp \Sigma w}$.

First, the portfolio loss is:
$$L(w) := -w\trsp X \overset{d}{=} -w\trsp \mu+Yv\trsp U\sqrt{w\trsp \Sigma w},$$
where $v\trsp = w\trsp\Sigma^{1/2}/\sqrt{w\trsp\Sigma w}$, with $||v||_2=1$.
Before we compute $CVaR(-w\trsp X;\beta)=CVaR(L(w);\beta)$, we need to compute $\alpha_\beta$, the VaR of $L(w)$ at level $\beta$ [equivalently, the $(1-\beta)$-quantile of $L(w)$].
Since $L(w)$ is a continuous random variable, $\alpha_\beta = F^{-1}_{L(w)}(1-\beta)$, where $F_{L(w)}^{-1}$ is the inverse cdf of $L(w)$. Now
\begin{eqnarray*}
F_{L(w)}(x) = P(L(w) \leq x)
= P\left(Yv\trsp U \geq \frac{-x-w\trsp \mu}{\sqrt{w\trsp \Sigma w}}\right),
\end{eqnarray*}
so to compute $\alpha_\beta$, we need the distribution of $Yv\trsp U$. Since $v$ has norm 1, $v\trsp Z_p \overset{d}{=} Z_1$, where $Z_1\sim\mathcal{N}(0,1)$, and
since $U\overset{d}{=}Z_p/||Z_p||_2$,
$$
v\trsp U \overset{d}{=} \frac{Z_1}{\sqrt{Z_1^2+\chi^2_{p-1}}}\;,
$$
where $\chi^2_{p-1}$ is independent of $Z_1$. Thus $(v\trsp U)^2 \sim Beta(1/2,(p-1)/2)$, and by the symmetry of the normal, we have
$$
P(Y v\trsp U\geq x)=P(Y I(1/2) \sqrt{B}\geq x)\;,
$$
where $B\sim Beta(1/2,(p-1)/2)$ and $I(1/2)\sim Bernoulli(1/2)$, independent of the rest.
This quantity clearly does not depend on our choice of $w$, hence the solution
to the equation
\begin{eqnarray*}
F_{L(w)}(x)= 1-\beta
\end{eqnarray*}
is given by
$$\alpha_\beta = -w\trsp \mu+q(1-\beta;Y I(1/2)  \sqrt{B})\sqrt{w\trsp \Sigma w},$$
where $q$ is a function that does not depend on $w$, and is unique since $L(w)$ is a continuous random variable.

Thus CVaR at level $\beta$ is given by
\begin{eqnarray}\label{eq:CVAR ellip}
CVaR(L(w);\beta) &=& \frac{1}{1-\beta}\Exp[L(w)I(L(w)\geq \alpha_\beta)]\nonumber\\
&=&-w\trsp \mu+G(1-\beta;YI(1/2)  \sqrt{B})\sqrt{w\trsp \Sigma w},
\end{eqnarray}
\noindent where $G$ does not depend on $w$. Hence minimizing $CVaR(L(w);\beta)$ subject
to $w\trsp \mu = R$ and $w\trsp 1_p = 1$ is equivalent to minimizing $w\trsp \Sigma w$ subject to the same constraints,
which is precisely the population Markowitz problem.
\end{proof}

\section{Theory}\label{sec:theory}
We have thus far introduced nonparametric and parametric PBR methods to improve upon the empirical mean-CVaR problem
(\ref{eq:empricial problem}).
While we evaluate these methods in Sec.~\ref{sec:numerical} via simulation experiments, it is still desirable to obtain some theoretical
understanding of $\hat{w}_n$, $\hat{w}^v_n$ and $\hat{w}^{M}_n$.

The solution to the empirical Markowitz problem $\hat{w}^{M}_n$ has an explicit form and its asymptotic behavior
has been studied elsewhere [for $X\sim\mathcal{N}(\mu,\Sigma)$, see \citeasnoun{jobson-korkie}, and for $X\sim Elliptical$, see
\citeasnoun{el2009high}].
So we focus on deriving the asymptotic behavior of $\hat{w}_n$ and $\hat{w}^v_n$ ---
specifically, we show that they follow the Central Limit Theorem (CLT).
Application of the delta method from classical statistics [see for example, Chapter 3 of \citeasnoun{vandervaart2000}]
then allows us to conclude that the corresponding sample efficient frontiers
also follow the CLT. From these results, we can get some insight into the effect of the penalty functions in the nonparametric PBR method,
and (indirectly) justify the parametric PBR method when the log-returns are Gaussian.

\textbf{Notations.} In this section, we make use of stochastic little-o and big-O notations:
for a given sequence of random variables $R_n$, $X_n = o_P(R_n)$ means
$X_n = Y_nR_n$ where $Y_n\overset{P}{\rightarrow}0$, and $X_n = O_P(R_n)$ means
$X_n = Y_nR_n$ where $Y_n=O_P(1)$, i.e. for every $\varepsilon>0$ there exists a constant $M$ such that $\underset{n}{\sup}~P(|Y_n|>M)<\varepsilon$.

\textbf{Measurability Issues.} We also encounter quantities that may not be measurable (e.g. supremum over uncountable families of measurable functions).
We note that whenever the ``probability'' of such quantities are written down, we actually mean the outer probability.
For further details, see Appendix C of \citeasnoun{pollard1984convergence}.

\subsection{Preliminaries}\label{subsec:Preliminaries}
The quantities $\hat{w}_n$ and $\hat{w}^v_n$ are solutions to non-trivial optimization problems so
they cannot be written down analytically, and it seems characterizing their asymptotic distributions would be difficult.
However, we are not at a complete loss. In statistics, an M-estimator\footnote{``M'' stands for Minimization (or Maximization). For readers unfamiliar with M-estimation, maximum likelihood estimation falls in this category.}
is an estimator that minimizes an empirical function of the type
\begin{eqnarray}\label{eq:M-estimation}
\theta\mapsto M_n(\theta):=\frac{1}{n}\sum_{i=1}^n m_{\theta}(X_i),
\end{eqnarray}
where $X_1,\ldots,X_n$ are iid observations, over some parameter space $\Theta$. The solution $\hat{\theta}_n$ is then
a reasonable estimator of the minimizer $\theta_0$ of the true mean $M(\theta) = \mbb{E}[m_{\theta}(X_1)]$.
It is well-known that $\hat{\theta}_n$ obeys the Central Limit Theorem (i.e.~is asymptotically normally distributed)
under some regularity conditions. Intuitively, assuming $\theta$ is one-dimensional and $M_n$ is sufficiently smooth, the CLT result is based on Taylor expansion of the first-order condition $dM_n(\hat{\theta}_n)/d\theta = 0$ about $\theta_0$:
\begin{eqnarray*}
0 &=& \frac{dM_n(\hat{\theta}_n)}{d\theta} =  \frac{dM_n(\theta_0)}{d\theta}+(\hat{\theta}_n-\theta_0)\frac{d^2M_n(\theta_0)}{d\theta^2}+O_P(|\hat{\theta}_n-\theta_0|^2).
\end{eqnarray*}
Under reasonable assumptions that $d^2M_n(\theta_0)/d\theta^2$ obeys the Weak Law of Large Numbers and $\hat{\theta}_n$ is a consistent estimator of $\theta_0$
(i.e.~$|\hat{\theta}_n-\theta_0|\overset{P}{\rightarrow}0$), we have
\begin{equation*}
\sqrt{n}(\hat{\theta}_n-\theta_0)
= -\frac{1}{\mbb{E}[\frac{d^2M_n(\theta_0)}{d\theta^2}]}\frac{1}{\sqrt{n}}\sum_{i=1}^n\frac{dm_{\theta_0}(X_i)}{d\theta}+o_P(1),
\end{equation*}
with the latter expression obeying the standard CLT as it is a normalized sum of iid random variables.

So we ask, can we transform (\ref{eq:empricial problem}) and (\ref{eq:penalized problem-asympvar}) to a problem for which we can use the M-estimation results?

The first step towards transforming (\ref{eq:empricial problem}) and (\ref{eq:penalized problem-asympvar}) is to make them into
constraint-free optimization problems. This is achievable, albeit with some thoughts, and we defer the details to Sec.~\ref{subsec:consistency}.
Next, we need to show $\hat{w}_n$ and $\hat{w}^v_n$ are consistent, i.e.~they converge in probability to
the corresponding population solutions. The proof of consistency is also provided in Sec.~\ref{subsec:consistency}.

Once (\ref{eq:empricial problem}) is transformed to a global optimization problem,
it is equivalent to an M-estimation problem in that the objective is a sample average of iid random variables of the form Eq.~(\ref{eq:M-estimation}).
Thus we conclude $\hat{w}_n$ is
asymptotically normally distributed with mean $w_0$ and covariance matrix
$\Sigma_{w_0}$, which we can compute.

However, (\ref{eq:penalized problem-asympvar}) after transformation into a global problem is not quite an M-estimation problem, because, after some algebra, the objective is of the form (see Sec.~\ref{subsec:consistency} for details):
\begin{equation}\label{eq:U-statistic}
\theta\mapsto M_n(\theta)=\frac{1}{n(n-1)}\sum_{i\neq j} m^U_{\theta}(X_i,X_j),
\end{equation}
where $m^U(\cdot,\cdot)$ is a permutation-symmetric function, and the sum is over all possible pairs $(i,j)$ for $1\leq i,j \leq n$, resulting in a sample average of identically distributed but non-independent terms.

For fixed $\theta$, statistics of the form Eq.~(\ref{eq:U-statistic}) are known as U-statistics, and we believe the solution $\hat{w}^v_n$
is still well-behaved because U-statistics can be decomposed into a term of the form $M^1_n(\theta)=\sum_{i=1}^nm^1_\theta(X_i)$
(known as its Haj\'{e}k projection or first term in its Hoeffding decomposition; see \citeasnoun{hoeffding1948class}) and a remainder
which converges to zero in probability at rate $\sqrt{n}$. Thus we intuit that the asymptotic behavior of $\hat{w}^v_n$ is equivalent to
the minimizer of $M^1_n(\theta)$, the latter for which we can apply the standard M-estimation result.
We make this intuition rigorous in Sec.~\ref{subsec:CLT for U}.
In Sec.~\ref{subsec:CLT for CVAR}, we provide details of the asymptotic distributions of $\hat{w}_n$ and $\hat{w}^v_n$
when $X\sim\mathcal{N}(\mu,\Sigma)$, and provide a justification of the parametric PBR method.

\subsection{Consistency of $\hat{w}_n$ and $\hat{w}^v_n$}\label{subsec:consistency}
In this subsection we show consistency of $\hat{w}^v_n=\hat{w}^v_n(\lambda_1,\lambda_2)$.
The result goes through for $\hat{w}_n$ by setting $\lambda_1=\lambda_2=0$.

\subsubsection{Transformation into global optimization}
The penalized CVaR portfolio optimization problem with dualized mean and sample/asymptotic variance penalty constraints is
\begin{equation}\tag{CVaR-dual}\label{eq:penalized problem-dual}
     \begin{array}{ll}
      \underset{(\alpha,w)\in\mbb{R}\times\mbb{R}^{p}}{\min} & M_n(\alpha,w;\lambda_1,\lambda_2) \\
      s.t. &~~~w\trsp 1_p = 1,
\end{array}
\end{equation}
where
\begin{equation}\label{eq:M_n}
M_{n}(\theta;\lambda_1,\lambda_2)= \frac{1}{n}\sum_{i=1}^nm_{\theta}(X_i)+\frac{\lambda_1}{n}w\trsp \hat{\Sigma}_nw+\frac{\lambda_2}{n-1}\sum_{i=1}^n\left( z_\theta(X_i)- \frac{1}{n}\sum_{j=1}^nz_\theta(X_j)\right)^2,
\end{equation}
\begin{equation}\label{eq:m_theta}
m_{\theta}(x) = \alpha+\frac{1}{1-\beta}z_\theta(x)-\lambda_0w\trsp x,
\end{equation}
and $\lambda_0>0$, $\lambda_1,\lambda_2\geq 0$ are pre-determined constants.

We dualize the mean constraint $w'\hat{\mu}_n = R$ because it makes the analysis of the corresponding solution much easier.
While dualizing the mean constraint adds a sample average of iid terms to the objective,
leaving it as a constraint results in a solution that has a non-trivial dependence on the underlying randomness.

Now eliminating the non-random constraint $w\trsp 1_p = 1$ is straight-forward; one possible way is to re-parameterize
$w$ as $w=w_1+Lv$, where $L=[0_{(p-1)\times 1},I_{(p-1)\times(p-1)}]\trsp$, $v = [w_2,\ldots,w_p]\trsp$
and $w_1=[1-v\trsp 1_{(p-1)},0_{1\times (p-1)}]\trsp$.
The transformed problem is thus
\begin{equation}
\underset{\theta\in\mbb{R}^p}{\min}~M_n(\theta;\lambda_1,\lambda_2),
\end{equation}
where $\theta=(\alpha,v)\in\mbb{R}\times\mbb{R}^{p-1}$ is free of constraints, and the corresponding population problem is
\begin{equation}\label{eq:M}
\underset{\theta\in\mbb{R}^p}{\min}~M(\theta;\lambda_1,\lambda_2)=\mbb{E}[M_n(\theta;\lambda_1,\lambda_2)].
\end{equation}

In what follows, we assume $M(\theta;\lambda_1,\lambda_2)$ has a unique minimizer $\theta_0(\lambda_1,\lambda_2)$.
We also let $\hat{\theta}_n(\lambda_1,\lambda_2)$ be a near-minimizer of $M_n(\theta;\lambda_1,\lambda_2)$, i.e.
\begin{equation}\label{eq:nearmin}
M_n(\hat{\theta}_n;\lambda_1,\lambda_2) < \inf_{\theta\in\mbb{R}^p}M_n(\theta;\lambda_1,\lambda_2)+o_P(1).
\end{equation}


\subsubsection{Transformation of the objective to a U-statistic}
Let $\theta=(\alpha,v)\in\mbb{R}\times\mathbb{R}^{p-1}$ and $z_\theta(x) := (-x\trsp (w_1+Lv)-\alpha)^+.$
With simple algebra, we can re-write the objective Eq.~(\ref{eq:M_n}) as a U-statistic:
\begin{equation}\label{eq:M_n U}
M_{n}(\theta;\lambda_1,\lambda_2) = \frac{1}{{n\choose 2}}\sum_{\substack{1\leq i,j \leq n \\ i\neq j}} m^U_{(\theta;\lambda_1,\lambda_2)}(X_{i},X_{j}),
\end{equation}
where
\begin{equation}\label{eq:mU}
m^U_{(\theta;\lambda_1,\lambda_2)}(x_i,x_j):=
\frac{1}{2}\left[m_{\theta}(x_i)+m_{\theta}(x_j)\right]
+ \frac{\lambda_1}{2}[(w_1+Lv)\trsp (x_i-x_j)]^2 +\frac{\lambda_2}{2}(z_\theta(x_i)-z_\theta(x_j))^2.
\end{equation}

\subsubsection{Consistency of $\hat{\theta}_{n}(\lambda_1,\lambda_2)$}
Let us now prove consistency of $\hat{\theta}_{n}(\lambda_1,\lambda_2)$ for fixed $\lambda_1,\lambda_2\geq 0$.
The intuition behind the proof is as follows: if $M(\theta;\lambda_1,\lambda_2)$ is well-behaved such that for every $\varepsilon>0$
there exists $\eta>0$ such that $||\hat{\theta}_n(\lambda_1,\lambda_2) - \theta_0(\lambda_1,\lambda_2)||_2>\varepsilon
\implies M(\hat{\theta}_n;\lambda_1,\lambda_2)-M(\theta_0;\lambda_1,\lambda_2)>\eta$,
then consistency follows from showing that the probability of the event
$\{M(\hat{\theta}_n;\lambda_1,\lambda_2)-M(\theta_0;\lambda_1,\lambda_2)>\eta\}$
goes to zero for all $\varepsilon>0$. In the proof, we show that
$0\leq M(\hat{\theta}_n;\lambda_1,\lambda_2)-M(\theta_0;\lambda_1,\lambda_2) \leq
-(M_n(\hat{\theta}_n;\lambda_1,\lambda_2)-M(\hat{\theta}_n;\lambda_1,\lambda_2))+o_P(1)$,
hence the result follows by proving Uniform Law of Large Numbers (ULLN) for $M_n(\theta;\lambda_1,\lambda_2)$:
\begin{equation}\label{eq:ULLN}
\sup_{\theta\in\mbb{R}^{p}}|M_n(\theta;\lambda_1,\lambda_2)-M(\theta;\lambda_1,\lambda_2)| \overset{P}{\rightarrow}0.
\end{equation}



ULLN has been extensively studied in the statistics and empirical processes literature and one of the standard
approaches to showing ULLN is through bracketing numbers. Given two functions $l,u$, the bracket $[l,u]$
is the set of all functions $g$ with $l\leq g\leq u$. An $\varepsilon$-bracket in $L_r(P)$ is a bracket $[l,u]$ with $\mbb{E}_{P}(u-l)^r<\varepsilon^{r}$, and the
bracketing number $N_{[~]}(\varepsilon,\mathcal{F},L_r(P))$ is the minimum number of $\varepsilon$-brackets
needed to cover $\mathcal{F}$. Having a finite bracketing number $N_{[~]}(\varepsilon,\mathcal{F},L_r(P))<\infty$
for every $\varepsilon>0$ means one can find a finite approximation to
$\mathcal{F}$ with $\varepsilon$-accuracy for all $\varepsilon>0$, and ULLN holds for such $\mathcal{F}$
[Theorem 19.4 \citeasnoun{vandervaart2000}].

There are certainly known sufficient conditions for finite bracketing numbers. For our problem, if we can replace
$\mbb{R}^{p}$ with a compact set,
we can show ${F}$ is a Lipschitz class of functions (defined in the next paragraph), which is known to have finite
$N_{[~]}(\varepsilon,\mathcal{F},L_r(P))$ for every $\varepsilon>0$. Now for all practical purposes, we need only
consider a compact subset of $\Theta$, $[-K,K]^{p}$ where $K$ is appropriately large enough, because the elements of $\theta=(\alpha,v)$
are only meaningful if bounded in size ($\alpha$ is the Value-at-Risk of the portfolio $w=w_1+Lv$). Hence for the rest of this section
we assume a $K$ exists such that $\hat{\theta}_n \in [-K,K]^{p}$ for all $n$ and $\theta_0\in[-K,K]^{p}$.

\begin{definition}[\textbf{Lipschitz class}]
Consider a class of measurable functions $\mathcal{F} = \{f_\theta: \theta\in\Theta\}$, $f_\theta:\mathcal{X}\rightarrow\mbb{R}$, under some probability measure $P$. We say $\mathcal{F}$ is a \emph{Lipschitz class} about $\theta_0\in\Theta$ if $\theta\mapsto f_\theta(x)$ is differentiable at $\theta_0$ for P-almost every $x$ with derivative $\dot{f}_{\theta_0}(x)$ and such that, for every $\theta_1$ and $\theta_2$ in a neighborhood of $\theta_0$, there exists a measurable function $\dot{f}$ with $\mathbb{E}[\dot{f}^2(X_1)]<\infty$ such that
\begin{eqnarray*}
|f_{\theta_1}(x)-f_{\theta_2}(x)|\leq \dot{f}(x)||\theta_1-\theta_2||_2.
\end{eqnarray*}
\end{definition}
Example 19.7 of \citeasnoun{vandervaart2000} shows that if $\mathcal{F} = \{f_\theta: \theta\in\Theta\}$ is a class of measurable functions with bounded $\Theta\subset\mbb{R}^d$ and $\mathcal{F}$ is Lipschitz about $\theta_0\in\Theta$ then for every $0<\varepsilon<diam(\Theta)$, there exists $C$ such that
\begin{equation}\label{eq:vdv 19.7}
N_{[~]}(\varepsilon\sqrt{\mbb{E}(|\dot{f}(X)|^2)},\mathcal{F},L_2(P))\leq  C\left(\frac{diam(\Theta)}{\varepsilon}\right)^d,
\end{equation}
i.e. has a finite bracketing number for all $\varepsilon>0$. This result is needed in proving consistency in the following.

\begin{theorem}\label{theo:consistency unreg}
For fixed $\lambda_1,\lambda_2\geq 0$, let $\hat{\theta}_n(\lambda_1,\lambda_2)$
be a near-minimizer of $M_{n}(\theta;\lambda_1,\lambda_2)$
as in Eq.~(\ref{eq:nearmin}), and let $\theta_0(\lambda_1,\lambda_2)$ be the unique minimizer of $M(\theta;\lambda_1,\lambda_2)$.
Also let
$$\mathcal{F}_1 = \{m_{\theta}:\theta\in [-K,K]^{p}\},~~~
\mathcal{F}_2 = \{m^U_{(\theta;\lambda_1,\lambda_2)}:\theta\in [-K,K]^{p}\},$$
where $m_{\theta}$ and $m^U_{(\theta;\lambda_1,\lambda_2)}$ are defined in Eqs.~(\ref{eq:m_theta}) and (\ref{eq:mU}).
Suppose the following:
\\\indent Assumption 1. $\theta\mapsto M(\theta;\lambda_1,\lambda_2)$ is continuous and
$\liminf_{|\theta|\rightarrow\pm\infty} M(\theta;\lambda_1,\lambda_2) > M(\theta_0;\lambda_1,\lambda_2)$.
\\\indent Assumption 2. $X_1,\ldots,X_n$ are iid continuous random vectors with finite fourth moment.
\\\noindent Then
$$||\hat{\theta}_n(\lambda_1,\lambda_2)-{\theta}_0(\lambda_1,\lambda_2)||_2\overset{P}{\rightarrow}0.$$
\end{theorem}
\begin{proof}
See Appendix \ref{App:const proof}.
\end{proof}

\subsection{Central Limit Theorem for $\hat{\theta}_n(\lambda_1,\lambda_2)$}\label{subsec:CLT for U}
We are now ready to show the CLT for $\hat{\theta}_n(\lambda_1,\lambda_2)$. The CLT for $\hat{\theta}_n(0,0)$ is a straightforward
application of known M-estimation results for Lipschitz class of objective functions [e.g. Theorem 5.23 of \citeasnoun{vandervaart2000}].

The CLT for $\hat{\theta}_n(\lambda_1,\lambda_2)$ when $\lambda_1,\lambda_2$ are not both zero does not follow straight-forwardly
from M-estimation
results because $M_n(\theta;\lambda_1,\lambda_2)$ is a sample average of identically distributed but non-independent terms.
However, statistics of the form $M_n(\theta;\lambda_1,\lambda_2)$ are known as U-statistics, and we can decompose them into a sum of iid random
variables and a component which is $o_P(1/\sqrt{n})$ [\citeasnoun{hoeffding1948class}]:
\begin{equation}\label{eq:Hoeff}
M_n(\theta;\lambda_1,\lambda_2) = \frac{1}{n}\sum_{i=1}^nm^1_{(\theta;\lambda_1,\lambda_2)}(X_i)+E_n(\theta;\lambda_1,\lambda_2),
\end{equation}
where $m^1_{(\theta;\lambda_1,\lambda_2)}(X_i) = 2\mbb{E}_{X_j}[m_{(\theta;\lambda_1,\lambda_2)}^U(X_i,X_j)]
-\mbb{E}_{X_1,X_2}[m^U_{(\theta;\lambda_1,\lambda_2)}(X_1,X_2)]$ and $E_n(\theta;\lambda_1,\lambda_2)=o_P(1/\sqrt{n})$.
Hence we suspect $|R^U_n(\hat{\theta}_n;\lambda_1,\lambda_2)|\overset{P}{\rightarrow}0$, where
\begin{eqnarray*}
R^U_n({\theta};\lambda_1,\lambda_2) = \sqrt{n}({\theta}-\theta_0)-[\nabla^2_{\theta_0}\mbb{E}m^1_{(\theta;\lambda_1,\lambda_2)}(X_i)]^{-1}
\frac{1}{\sqrt{n}}\sum_{i=1}^nm^1_{(\theta;\lambda_1,\lambda_2)}(X_i).
\end{eqnarray*}
Now $\hat{\theta}_n$ changes with every $n$ so we need uniform probabilistic
convergence of $R^U_n(\theta;\lambda_1,\lambda_2)$, and implicitly of $E_n(\theta;\lambda_1,\lambda_2)$.
For this we need to show convergence of particular stochastic processes; an empirical process and a U-process.

\begin{definition} Let $X_1,\ldots,X_n$ be iid random vectors from $\mathcal{X}$. For a measurable function $f:\mathcal{X}\rightarrow\mbb{R}$, the \emph{empirical process at $f$} is
$$\mbb{G}_nf : = \frac{1}{\sqrt{n}}\sum_{i=1}^n[f(X_i)-\mbb{E}f(X_1)],$$
and for a measurable function $g:\mathcal{X}\times\mathcal{X}\rightarrow\mbb{R}$, the \emph{U-process at $g$} is
$$\mbb{U}_ng : = \frac{\sqrt{n}}{{n\choose 2}}\sum_{i\neq j}[g(X_i,X_j)-\mbb{E}_{X_1,X_2}g(X_1,X_2)].$$
\end{definition}

To show convergence of quantities such as $\sup_{t\in T}|X_{n}(t)|$ for some stochastic process $\{X_n(t):t\in T\}$, we need to introduce the
notion of weak convergence of stochastic processes. If $X_n(\cdot,\omega)$ is a bounded function for every $\omega\in\Omega$, then we can consider $X_{n}(\cdot,\omega)$ to be a point in the function space $\ell^\infty({T})$, the space of bounded functions on $T$ which is equipped with the supremum norm.
Hence, showing the convergence of $\sup_{t\in T}|X_{n}(t)|$ is equivalent to showing weak convergence of $X_n$ in this function space.

\begin{definition}[\textbf{Weak convergence of a stochastic process}]
A sequence of $X_n:\Omega_n\mapsto\ell^\infty(T)$ converges weakly to a tight random element
\footnote{A random element is a generalization of a random variable. Let $(\Omega,\mathcal{G},P)$ be a probability space and $\mbb{D}$ a metric space.
Then the $\mathcal{G}$-measurable map $X:\Omega\mapsto \mbb{D}$ is called a random element.} $X$ iff both of the following conditions hold:
\begin{enumerate}
\item Finite approximation: the sequence $(X_n(t_1),\ldots,X_n(t_k))$ converges in distribution in $\mbb{R}^k$ for every finite set of points $t_1,\ldots,t_k$ in $T$.
\item Maximal inequality: for every $\varepsilon,\eta>0$ there exists a partition of $T$ into finitely many sets $T_1,\ldots,T_k$ such that
\begin{equation*}
\underset{n\rightarrow\infty}{\lim\sup}~P\left[\underset{i}{\sup}\underset{s,t\in T_i}{\sup}|X_{n}(s)-X_{n}(t)|\geq\varepsilon\right]\leq\eta.
\end{equation*}
\end{enumerate}
\end{definition}
The point at the end of this is, as taking the supremum is a continuous map in the topology of $\ell^\infty(T)$, weak convergence of $X_n(\cdot)$ to $X(\cdot)$ would allow us to conclude $\sup_{t\in T}|X_{n}(t)|\rightarrow \sup_{t\in T}|X(t)|$.

Regarding empirical processes, we say a class of measurable functions $\mathcal{F}$ is \emph{P-Donsker} if $\{\mbb{G}_nf:f\in\mathcal{F}\}$ converges weakly to a tight random element in $\ell^\infty(\mathcal{F})$. This property is related to the bracketing numbers introduced in Sec.~\ref{subsec:consistency}:
a class $\mathcal{F}$ is P-Donsker if $\varepsilon \log[N_{[~]}(\varepsilon,\mathcal{F},L_2(P))]\rightarrow0$
as $\varepsilon\rightarrow0$ [due to Donsker; see Theorem 19.5 of \citeasnoun{vandervaart2000}]. Many sufficient conditions for
the weak convergence of $\{\mbb{U}_nf:f\in\mathcal{F}\}$ are provided in \citeasnoun{arcones1993limit}, and we make use
of one in our proof of CLT for $\hat{\theta}_n(\lambda_1,\lambda_2)$ below.

\begin{theorem}\label{theo:Uasympt}
Fix $\lambda_1,\lambda_2\geq0$, $\lambda_1,\lambda_2$ not both zero and assume the same setting as Theorem~\ref{theo:consistency unreg}.
Also let
$$\dot{m}^U_{(\theta_0;\lambda_1,\lambda_2)}(x) = \nabla_\theta m^U_{(\theta_0;\lambda_1,\lambda_2)}(x) |_{\theta=\theta_0(\lambda_1,\lambda_2)},~\text{for}~x\in\mbb{R}^p,$$
and further assume
\\\indent Assumption 3. $\mbb{E}_{X_1,X_2}[m^U_{(\theta_0;\lambda_1,\lambda_2)}(X_1,X_2)^2]<\infty$.
\\\indent Assumption 4. $\theta\mapsto M(\theta;\lambda_1,\lambda_2)$ admits a second-order Taylor expansion at its point of minimum $\theta_0(\lambda_1,\lambda_2)$ with
nonsingular symmetric second derivative matrix $V_{\theta_0(\lambda_1,\lambda_2)}$.
\\\noindent Then
\begin{eqnarray*}
\sqrt{n}(\hat{\theta}_n(\lambda_1,\lambda_2)-\theta_0(\lambda_1,\lambda_2))
&=&
-V_{\theta_0(\lambda_1,\lambda_2)}^{-1}\frac{1}{\sqrt{n}}\sum_{i=1}^n\dot{m}^1_{(\theta_0;\lambda_1,\lambda_2)}(X_i)+o_{p}(1)
\end{eqnarray*}
where
$$\dot{m}^1_{(\theta;\lambda_1,\lambda_2)}(X_i) = 2\mbb{E}_{X_2}[\dot{m}^U_{(\theta;\lambda_1,\lambda_2)}(X_1,X_2)]
-\mbb{E}_{X_1,X_2}[\dot{m}^U_{(\theta;;\lambda_1,\lambda_2)}(X_1,X_2)]$$
is the first-order term in the Hoeffding decomposition of $M_n(\theta;\lambda_1,\lambda_2)$.
\end{theorem}

\textbf{Remark -- Implication on the choice of penalty functions.}
\\We have just shown that asymptotically, the sample variance penalty functions affect the solution performance only through its Haj\'{e}k projection.
This observation can generalize to many typical penalty functions (e.g. different statistics of mean and CVaR estimators), and as such, the implication
is that of all possible penalty functions to consider, one may focus on a subclass of functions that can be expressed as a sample average of iid terms.

\begin{corollary}\label{cor:asymptotic reg soln}
Assume the same setting as Theorem \ref{theo:Uasympt}. Then
\begin{eqnarray}\label{eq:asympt normality empr soln}
\sqrt{n}(\hat\theta_n(\lambda_1,\lambda_2) - \theta_0(\lambda_1,\lambda_2)) \Rightarrow
\mathcal{N}\left(0,\Sigma_{\theta_0}(\lambda_1,\lambda_2)\right),
\end{eqnarray}
where $\Sigma_{\theta_0}(\lambda_1,\lambda_2) = A^{-1}BA^{-1}$,
\begin{eqnarray*}
A
= A_{\theta_0}(\lambda_1,\lambda_2) &=&  \nabla^2_\theta~\mathbb{E}[m^1_{(\theta;\lambda_1,\lambda_2)}(X_1)]\bigg|_{\theta=\theta_0}\\
&=& \nabla^2_\theta~[\frac{1}{1-\beta}\mbb{E}z_\theta(X_1)+\lambda_1w\trsp\Sigma w+\lambda_2Var(z_\theta(X_1))]\bigg|_{\theta=\theta_0}\\
B = B_{\theta_0}(\lambda_1,\lambda_2) &=&
\mathbb{E}[\nabla_{\theta_0}{m}^1_{(\theta;\lambda_1,\lambda_2)}(X_1)\nabla_{\theta_0}{m}^1_{(\theta;\lambda_1,\lambda_2)}(X_1)\trsp]
\end{eqnarray*}
where
\begin{eqnarray*}
&&\nabla_{\theta}~m^1_{(\theta;\lambda_1,\lambda_2)}(x)\\
&&=\left[\begin{array}{c}
      1-\frac{1}{1-\beta}\mbb{I} +2\lambda_2\mbb{E}[(z_\theta(X)-\mbb{E}z_\theta(X))(-\mbb{I}+\mbb{E}\mbb{I})] \\
     -\frac{1}{1-\beta}L\trsp X\mbb{I}-\lambda_0L\trsp X+2\lambda_1L\trsp(X-\mu)(X-\mu)\trsp w
     +2\lambda_2 \mbb{E}[(z_\theta(X)-\mbb{E}z_\theta(X))(-L\trsp X\mbb{I}+\mbb{E}L\trsp X\mbb{I})]
     \\
    \end{array}  \right],
\end{eqnarray*}
and $\mbb{I} = \mbb{I}(z_\theta(X)\geq 0)$.
\end{corollary}

\textbf{Remarks.}
\begin{enumerate}
\item For asymptotics of $\hat{w}_n(\lambda_1,\lambda_2)$ we have
\begin{eqnarray}\label{eq:hatwn}
\sqrt{n}(\hat{w}_n(\lambda_1,\lambda_2) - w_{0}(\lambda_1,\lambda_2)) \Rightarrow
\mathcal{N}\left(0,\Sigma_{w_{0}}(\lambda_1,\lambda_2)\right),
\end{eqnarray}
where $\Sigma_{w_{0}}(\lambda_1,\lambda_2)=({0}_{p}~ L)\Sigma_{\theta_0}(\lambda_1,\lambda_2)({0}_{p}~ L)^\top$.
\item Setting $\lambda_1,\lambda_2=0$, we get back the unpenalized mean-CVaR problem.

\item \textbf{Asymptotic distribution of the efficient frontier.}
\\With Eq.~(\ref{eq:hatwn}), we can state the distribution of the true efficient frontier ---
that is, the distribution of $\hat{w}_n(\lambda_1,\lambda_2)\trsp\mu$ and $g(\hat{w}_n(\lambda_1,\lambda_2)):=
CVaR(-\hat{w}_n(\lambda_1,\lambda_2)\trsp X_{n+1};\beta)$,
where $X_{n+1}\sim F$, independent
of $X_1\ldots,X_n$.
For the portfolio mean, we have
\begin{eqnarray*}
\sqrt{n}(\hat{w}_n(\lambda_1,\lambda_2)\trsp\mu-w_0(\lambda_1,\lambda_2)\trsp\mu)\Rightarrow \mathcal{N}(0,\mu^\top\Sigma_{w_{0}}((\lambda_1,\lambda_2))\mu)
\end{eqnarray*}
and for the true CVaR, by the delta Method
\begin{equation}\label{eq:w'mu asympt}
\sqrt{n}(g(\hat{w}_n(\lambda_1,\lambda_2))-g(w_0(\lambda_1,\lambda_2)))\Rightarrow
\mathcal{N}\left\{0,g'(w_0(\lambda_1,\lambda_2))^\top\Sigma_{w_{0}}(\lambda_1,\lambda_2) g'(w_0(\lambda_1,\lambda_2))\right\}.
\end{equation}

The asymptotic distribution of $g(\hat{w}_n(\lambda_1,\lambda_2))$ clearly depends on the distribution of the assets $X$.
In the case when $X\sim ~Ellip(\mu,\Sigma,Y)$, $g(w) = -w^\top\mu+G\sqrt{w\trsp \Sigma w}$ according to our previous calculations
in Eq.~(\ref{eq:CVAR ellip}). Hence
\begin{equation}\label{eq:risk asympt}
\sqrt{n}(g(\hat{w}_n)-g(w_0))
\Rightarrow \mathcal{N}\left(0,\left(-\mu+G\frac{\Sigma w_0}{\sqrt{w_0\Sigma w_0}}
\right)^\top\Sigma_{w_0} \left(
-\mu+G\frac{\Sigma w_0}{\sqrt{w_0\Sigma w_0}}\right)\right).
\end{equation}
\end{enumerate}

\subsection{Example. Asymptotic analysis for $X\sim \mathcal{N}(\mu,\Sigma)$}\label{subsec:CLT for CVAR}
In the following, we provide the detailed computation of $\Sigma_{\theta_0}(0,0)$ for the unpenalized solution $\hat{\theta}_n(0,0)$
when $X\sim \mathcal{N}(\mu,\Sigma)$.

\begin{lemma}Suppose $X\sim \mathcal{N}(\mu,\Sigma)$. Then
\begin{eqnarray*}
z_{\theta_0}(X)&=&-w_{0}\trsp X-\alpha_{0}\sim\sigma_{0}\mathcal{N}(-\Phi^{-1}(\beta),1),~\text{and}\\
p_0&=&f_{-w_{0}\trsp X}(0)=\frac{1}{\sqrt{2\pi}\sigma_{0}}\exp\left\{-\frac{1}{2\sigma^2_{0}}(\Phi^{-1}(\beta))^2\right\},
\end{eqnarray*}
where $\sigma_{0}=\sqrt{w_{0}\trsp \Sigma w_{0}}$. Then $\Sigma_{\theta_0}(0,0) = A_0^{-1}B_0A_{0}^{-1}$,
where $A_0$, $B_0$ are symmetric matrices with
\begin{eqnarray*}
\begin{array}{lll}
A_0(1,1) & =\displaystyle\frac{p_0}{1-\beta} &  \\
A_0(j,l)&=\displaystyle\frac{p_0}{(1-\beta)}\mbb{E}[L_j\trsp X L_l\trsp X|z_{\theta_0}(X)=0]&\text{for}~ 2\leq j,l\leq p\\
A_0(1,j)&=\displaystyle-\frac{p_0}{(1-\beta)}\mbb{E}[L_j\trsp X|z_{\theta_0}(X)=0] &\text{for}~ 2\leq  j \leq p,
\end{array}
\end{eqnarray*}
where $L_j$ is the $j$-th column of $L$, and
\begin{eqnarray*}
\begin{array}{lll}
B_0(1,1)&=\displaystyle\frac{\beta}{1-\beta}&\\
B_0(j,l)&=\displaystyle\lambda_0^2(L_j\trsp\Sigma L_l+L_j\trsp\mu L_l\trsp\mu)+\frac{1}{(1-\beta)}\left(\frac{1}{1-\beta}+2\lambda_0\right)
\Exp[L_j\trsp X L_l\trsp X \mbb{I}(z_{\theta_0}(X)\geq 0)]& \text{for}~ 2\leq j,l\leq p\\
B_0(1,j)&= 0&\text{for}~2\leq j\leq p.
\end{array}
\end{eqnarray*}
\end{lemma}

\begin{proof}
This is a straight-forward application of Corollary \ref{cor:asymptotic reg soln} for the case $X\sim\mathcal{N}(\mu,\Sigma)$.
\end{proof}

Let us now compare the asymptotic results derived above with simulations with finite number of observations. Consider 5 assets,
a range of observations ($n=250,500,1000,2000$) and $X\sim\mathcal{N}(\mu_{sim},\Sigma_{sim})$, where the model parameters are the same as the model parameters of the first five assets used in Sec.~\ref{subsec:demo}.
For simulations, we solve the mean-CVaR problem with dualized mean constraint:
\begin{equation*}\label{eq:dual meanCVAR}
\begin{array}{ll}
\underset{w}{\min} & \widehat{CVaR}_n(-w\trsp \mathsf{X};\beta)-\lambda_0w\trsp \hat{\mu}_n  \\
s.t. & w\trsp 1_p = 1,
\end{array}
\end{equation*}
and follow steps similar to Sec.~\ref{subsec:demo}.

\begin{figure}[t!]
\centering
\includegraphics[width=0.7\paperwidth]{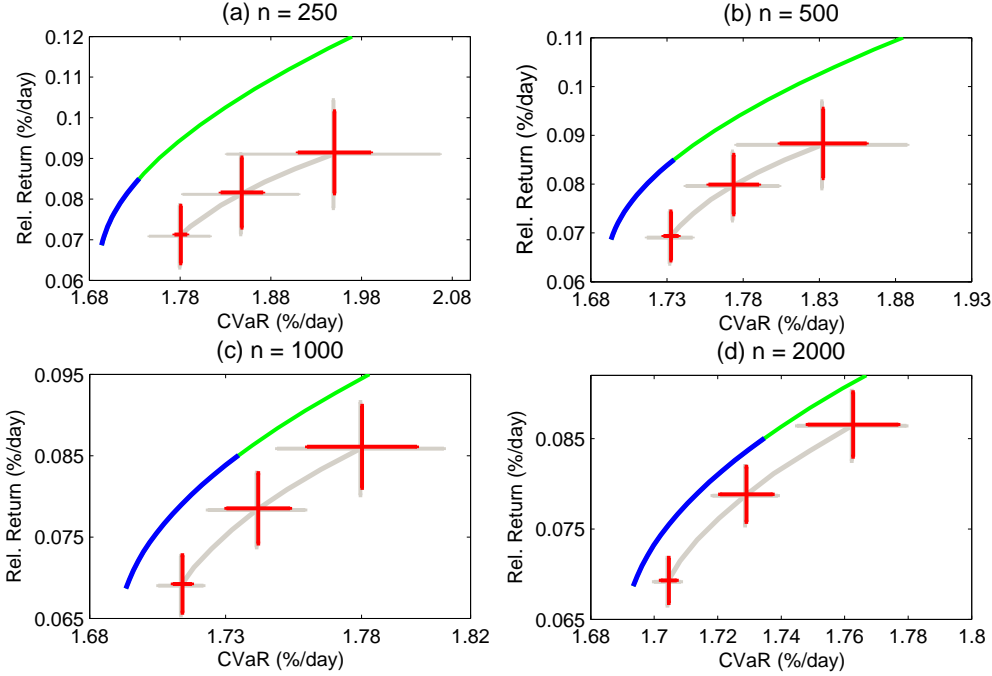}
 \caption{\small{Comparison of theoretical (red) and simulated (grey) distributions of the empirical efficient frontier
 when $X\sim\mathcal{N}(\mu_{sim},\Sigma_{sim})$ for increasing number of observations
 $n = [250,500,1000,2000]$. The error bars indicate $\pm 1/2$ std variabilities in the mean and CVaR.
 Green is the population efficient frontier, and blue indicates the portion that corresponds to the return range
 considered for the simulations. Observe that the asymptotic variance calculated theoretically (red bars) approach
 the simulated variance (grey bars) with increasing $n$.}}\label{fig:TheoSim}
\end{figure}

In Fig.~\ref{fig:TheoSim}, we summarize the empirical frontiers by plotting their averages and indicating $\pm 1/2$ standard deviation error bars,
in both true mean (vertical) and true risk estimations (horizontal) in grey. The population frontier is also plotted, and is shown in green,
and the theoretical $\pm 1/2$ standard deviations of mean and risk estimations are juxtaposed with the empirical error bars in red.
We make a couple of observations:
\begin{enumerate}
\item With increasing $n$, the theoretical error bars approach the simulated ones, as expected.
\item The theory seems to better predict the mean estimation error (vertical) better than the risk estimation error (horizontal).
With finite $n$, the mean estimation error, which is computed using Eq.~(\ref{eq:w'mu asympt}), depends only on one
approximate quantity $\Sigma_{w_0}(0,0)$, whereas the risk estimation error, computed using Eq.~(\ref{eq:risk asympt}),
depends on $\Sigma_{w_0}(0,0)$ and $w_0$. Although $\hat{w}_n$ is a consistent estimator of $w_0$ asymptotically,
with finite $n$ the difference does play a role, as shown by the relative inaccuracy of the horizontal error bars
compared to the vertical ones. The finite sample bias also explains the gap in the positions of the population and
simulated efficient frontiers.
\end{enumerate}

Let us now derive asymptotic properties of the penalized solution $\hat{\theta}_n(\lambda_1,\lambda_2)$, $\lambda_1,\lambda_2\geq 0$,
when $X\sim \mathcal{N}(\mu,\Sigma)$. First, we show that when $X\sim \mathcal{N}(\mu,\Sigma)$, penalizing variance of CVaR estimation is
redundant if one penalizes the sample variance of the mean.

\begin{lemma}
Suppose $X\sim\mathcal{N}(\mu,\Sigma)$ and let $z_{\theta}(X) = -\alpha-w\trsp X$. Then $z_{\theta}(X)\sim \mathcal{N}(\mu_1,\sigma_1^2)$
where $\mu_1 = -\sigma_1\Phi^{-1}(\beta)$, $\sigma_1^2 = w\trsp \Sigma w$, and
$$Var[\max(z_{\theta}(X),0)] = C(\beta)\sigma_1^2,$$
where $C(\beta)$ is a constant that only depends on $\beta$. Thus penalizing the sample variance of CVaR via
$P_2(w) = \widehat{Var}_n[z_{\theta(w)}(X),0)]  \leq U_2$ is redundant
if one penalizes the sample variance of the mean via $P_1(w) = w\trsp \hat{\Sigma}_n w = \hat{\sigma}_{1,n}^2 \leq U_1$.
\end{lemma}

\begin{proof}Straight-forward calculations show
\begin{eqnarray*}
Var[\max(z_{\theta}(X),0)] = \left\{([\Phi^{-1}(\beta)]^2+1)(1-\beta)-3\Phi^{-1}(\beta)f_{Z_0}[\Phi^{-1}(\beta)]\right\}\sigma_1^2,
\end{eqnarray*}
where $f_{Z_0}$ is the pdf of the standard normal random variable $Z_0$.
\end{proof}

The implication now is that when $X\sim \mathcal{N}(\mu,\Sigma)$, we need only consider $\lambda_1\geq 0,\lambda_2=0$ to characterize
the asymptotic properties of the penalized solution, which we describe below.

\begin{lemma}\label{lem:pen ast}Suppose $X\sim \mathcal{N}(\mu,\Sigma)$. Then
$$\Sigma_{\theta_0}(\lambda_1,0) = A_1^{-1}B_1A_{1}^{-1},$$
where $A_1$, $B_1$ are symmetric matrices with
\begin{eqnarray*}
\begin{array}{lll}
A_1(1,1)&=A_0(1,1)&\\
A_1(j,l)&=A_0(j,l)+\lambda_1L_j\trsp\Sigma L_l& \text{for}~ 2\leq j,l\leq p\\
A_1(1,j)&=A_0(1,j) &\text{for}~ 2\leq  j \leq p
\end{array}
\end{eqnarray*}
where $L_j$ is the $j$-th column of $L$, and
\begin{eqnarray*}
\begin{array}{lll}
B_1(1,1)&=B_0(1,1)&\\
B_1(j,l)&=B_0(j,l)+\lambda_1\mbb{E}[b_{0,j}b_{1,l}+b_{0,l}b_{1,j}+\lambda_1b_{1,j}b_{1,l}]&\text{for}~ 2\leq j,l\leq p\\
B_1(1,j)&=B_0(1,j) + \lambda_1\mbb{E}[b_{0,1}b_{1,j}] &\text{for}~2\leq j\leq p
\end{array}
\end{eqnarray*}
where for $2\leq j,l\leq p$,
\begin{eqnarray*}
\mbb{E}[b_{0,j}b_{1,l}]
&=& -\frac{2}{1-\beta}\mbb{E}[L_j\trsp X L_l\trsp(X-\mu)w\trsp(X-\mu)\mbb{I}(z_{\theta_0}(X)\geq0)] -2\lambda_0L_j\trsp\mu L_l\trsp \Sigma w\\
\mbb{E}[b_{1,j}b_{1,l}] &=& 4\mbb{E}[L_j\trsp(X-\mu)(X-\mu)\trsp L_lw\trsp(X-\mu)(X-\mu)\trsp w]\\
\mbb{E}[b_{0,1}b_{1,l}] &=& 2L_l\Sigma w-\frac{2}{1-\beta}\mbb{E}[ L_l\trsp(X-\mu)w\trsp(X-\mu)\mbb{I}(z_{\theta_0}(X)\geq0)].
\end{eqnarray*}
\end{lemma}

\begin{proof}
This is a straight-forward application of Corollary \ref{cor:asymptotic reg soln} for the case $X\sim\mathcal{N}(\mu,\Sigma)$.
\end{proof}

\textbf{Remark -- Justification of the parametric PBR method.}\\
The nonparametric PBR method with only a penalty on the mean estimation is a linear combination of the empirical mean-CVaR problem
(\ref{eq:empricial problem}) and the empirical Markowitz problem (\ref{eq:Markowitz}) because the penalty is precisely
the portfolio variance estimate $w\trsp \hat{\Sigma}_n w$. In particular, this single-penalty problem
approaches (\ref{eq:Markowitz}) with increasing $\lambda_1$. In Figure~\ref{fig:TheoErrorReg}, we plot 1 std of
$w^v_n(\lambda_1,0)\trsp \mu$ and $CVaR(-w^v_n(\lambda_1,0)\trsp X;\beta)$
for the single-penalty problem as $\lambda_1$ is increased, for different values of $\lambda_0$, computed using Lemma~\ref{lem:pen ast}.
Observe that the asymptotic standard deviations for both portfolio mean and CVaR decrease with increasing
$\lambda_1$, uniformly in $\lambda_0$. Given that both solutions to (\ref{eq:empricial problem}) and (\ref{eq:Markowitz}) converge to the
population solution $w_0$, the asymptotic theory deems the empirical Markowitz solution superior.

\begin{figure}[t!]
\centering
\includegraphics[width=0.65\paperwidth]{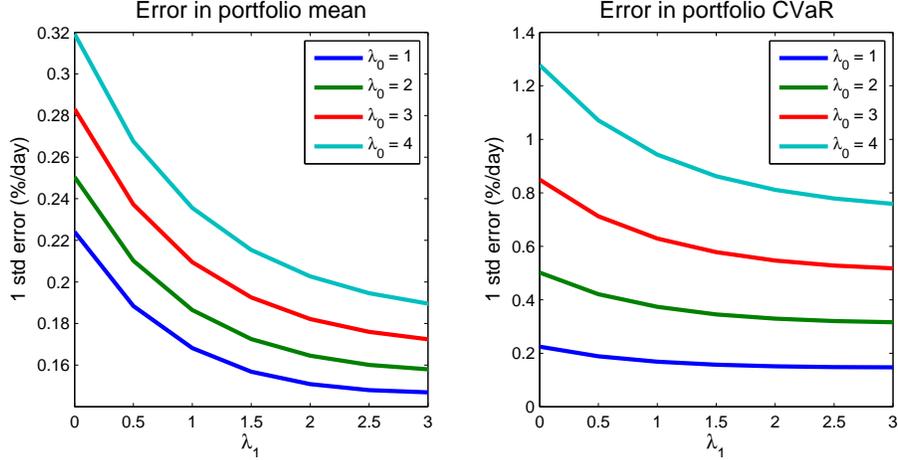}
 \caption{\small{1 asymptotic std of the portfolio mean and CVaR
for the single-penalty problem as $\lambda_1$ is increased when $X\sim\mathcal{N}(\mu,\Sigma)$,
for different values of $\lambda_0$.}}\label{fig:TheoErrorReg}
\end{figure}

\section{Numerical results}\label{sec:numerical}
In this section, we present simulation results to evaluate the nonparametric and parametric PBR methods presented in Sec.~\ref{sec:penalized problem}
against the straight-forward approach (\ref{eq:empricial problem}).
We consider $p = 10$ assets and three distributional models for the asset log-returns:
$X$ is multivariate Gaussian, elliptical and mixture of multivariate Gaussian and negative exponential.
For each model, we follow the procedure outlined in Section 2 to construct sample efficient frontiers
corresponding to (\ref{eq:empricial problem}), (\ref{eq:penalized problem-asympvar}) and (\ref{eq:Markowitz}).

One question that arises while solving (\ref{eq:penalized problem-asympvar}) is how one chooses the penalty terms $U_1$ and $U_2$ in the constraints
\begin{eqnarray*}
         \displaystyle\frac{1}{n}w\trsp \hat\Sigma_nw  &\leq& U_1 \\
         \displaystyle\frac{1}{n(1-\beta)^2}{z}\trsp \Omega_n{z} &\leq& U_2.
\end{eqnarray*}
If $U_1,U_2$ are too small, the problem becomes infeasible, whereas if they are too large, the penalization does not have any effect.
It is sensible to choose $U_1, U_2$ as a proportion of $\hat{w}_n\trsp \hat\Sigma_n\hat{w}_n/n$ and
$\hat{{z}}_n\trsp \Omega_n\hat{{z}}_n/(n(1-\beta)^2)$ respectively,
where $(\hat{w}_n,\hat{{z}}_n)$ is the solution to the unpenalized problem (\ref{eq:empricial problem}).
We denote the proportions $r_1$ and $r_2$ respectively. In practice, one would perform cross-validation to find
values of $(r_1,r_2)\in[0,1]\times[0,1]$ that maximize out-of-sample performance.

\subsection{Gaussian/elliptical models}
Here we consider
$$
X\sim \mu_{sim} +\lambda\mathcal{N}(0,\Sigma_{sim})
$$
where $\lambda$ is as in (\ref{eq:elliptical model}), with $\lambda=1$ for a Gaussian model and $\lambda \sim \Gamma(3,0.5)$ for an elliptical model. The parameters
$\mu_{sim}$ and $\Sigma_{sim}$ are the same as those used in Sec.~\ref{subsec:demo}.
We plot the histograms for $100,000$ sample returns for an equally-weighted portfolio $w = 1_p/p$
under the Gaussian and elliptical models in Fig. (\ref{fig:hist gauss ellip}).

We summarize the simulation results in Fig.~(\ref{fig:gauss_ellip}), where $(r_1,r_2) = (0.92,1)$ for
both the Gaussian and elliptical models [recall that the second penalty is redundant due to Lemma \ref{lem:equiv}].
Notice that for both models, the empirical Markowitz efficient frontier dominates the
penalized efficient frontier which in turn dominates the empirical mean-CVaR efficient frontier, in both \emph{position} of
the average of the simulated frontiers and \emph{variability}, as indicated by the vertical and horizontal error bars.

For the Gaussian case, $r_1=0.92$ was just feasible in that further reduction in this value led to most
instances of the problem being infeasible. From Fig.~(\ref{fig:gauss_ellip}b), we can see that this
is because the penalized solutions are approaching the empirical Markowitz solutions
with this choice of $r_1$ as the average simulated efficient frontiers of penalized (grey) and
empirical Markowitz (blue) solutions are close. For the elliptical model,
$r_1=0.92$ could be further reduced with the resulting penalized efficient frontier approaching the empirical Markowitz
efficient frontier. In summary, the empirical Markowitz solutions perform uniformly better than both the
original and penalized mean-CVaR solutions, with the penalized efficient frontier nearing the
empirical Markowitz efficient frontier with decreasing $r_1$.

\subsection{Mixture model}
Let us now consider returns being driven by a mixture of multivariate
normal and negative exponential distributions, such that with a small probability,
all assets undergo a perfectly correlated exponential-tail loss. Formally,
\begin{equation}
{X}\sim(1-I(q))N({\mu_{sim}},\Sigma_{sim})+I(q)(Y{1}_p+{f}),
\end{equation}
where $({\mu_{sim}},\Sigma_{sim})$ are parameters with the same value as in the Gaussian/elliptical models, $I(q)\sim Bernoulli(q)$, and ${f}=[f_{1},\ldots,f_{p}]\trsp $ is a $p\times1$ vector of constants, and
Y is a negative exponential random variable with density
\[
P(Y=y)=\begin{cases}
\lambda e^{\lambda y},~~ & \mathrm{if}~y\leq 0\\
0~~ & \mathrm{otherwise}.\end{cases}\]
In our simulations, we consider $q=0.05$, $f_{i}=\mu_{i}-\sqrt{\Sigma_{ii}}$ for $i=1,\ldots,p$ and $\lambda=1$.
The histogram for $100,000$ sample returns of an equally-weighted portfolio under this mixture model is shown in Fig. (4a).

We summarize the simulation results in Fig.~(4b), where $(r_1,r_2) = (0.5,0.5)$.
In this case, the penalized efficient frontiers perform better on average than
the efficient frontiers generated by the other two methods. The empirical Markowitz efficient frontiers do not seem to perform any
better than the original efficient frontiers on average, which is not surprising  because the empirical Markowitz solution is only
intended for $X$ having an elliptical distribution.

\section{Conclusion}\label{sec:conclusion}
We investigate Performance-Based Regularization as a method to reduce estimation risk in empirical
mean-CVaR portfolio optimization. The nonparametric PBR method solves the empirical mean-CVaR problem
with penalties on the uncertainties in mean and CVaR estimations. The parametric PBR method solves the empirical Markowitz
problem instead if the underlying model is elliptically distributed. Both theoretical analysis and simulation
experiments show the PBR methods improve upon the naive approach to data-driven mean-CVaR portfolio optimization.

From a larger perspective, the PBR approach is a new and promising way of
dealing with estimation risk and introducing robustness to data-driven optimization,
and is not restricted to the mean-CVaR problem. We leave investigating PBR in a general problem context for future work.

\section*{Acknowledgements}{This research was supported in part by the NSF CAREER Awards DMS-0847647 (El Karoui),
CMMI-0348746 (Lim), and NSF Grants CMMI-1031637 and CMMI-1201085 (Lim). The opinions, findings, conclusions or recommendations expressed
in this material are those of the authors and do not necessarily reflect the views of the National Science Foundation.
The authors also acknowledge support from an Alfred P. Sloan Research Fellowship (El Karoui), the Coleman Fung Chair in
Financial Modeling and the Coleman Fung Risk Management Center (Lim) and the Eleanor Sophia Wood Traveling Scholarship from
The University of Sydney (Vahn).}

\begin{figure}[htbp]
\centering
\includegraphics[width=0.7\paperwidth]{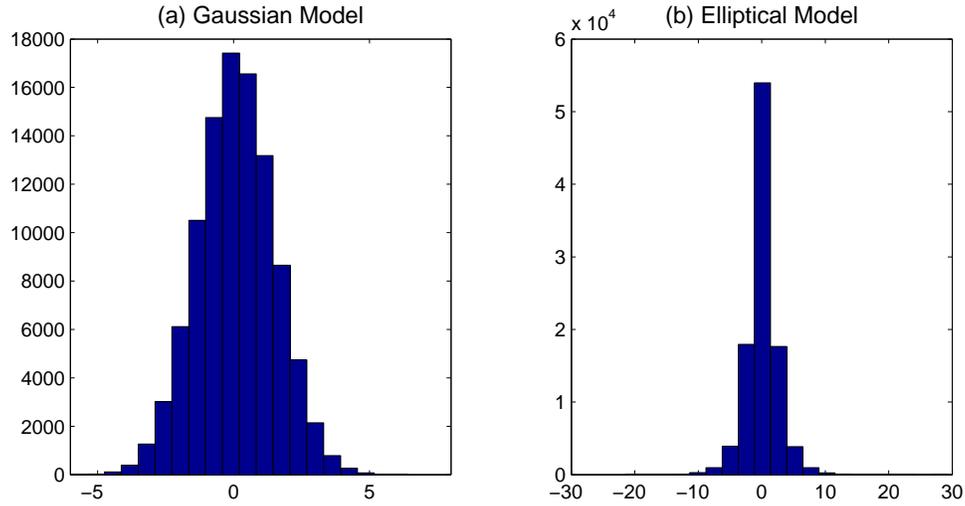}
\caption{\small Distribution of equally weighted portfolio under (a) Gaussian and (b) elliptical model}\label{fig:hist gauss ellip}
\end{figure}

\begin{figure}[h!]
\includegraphics[width=0.8\paperwidth]{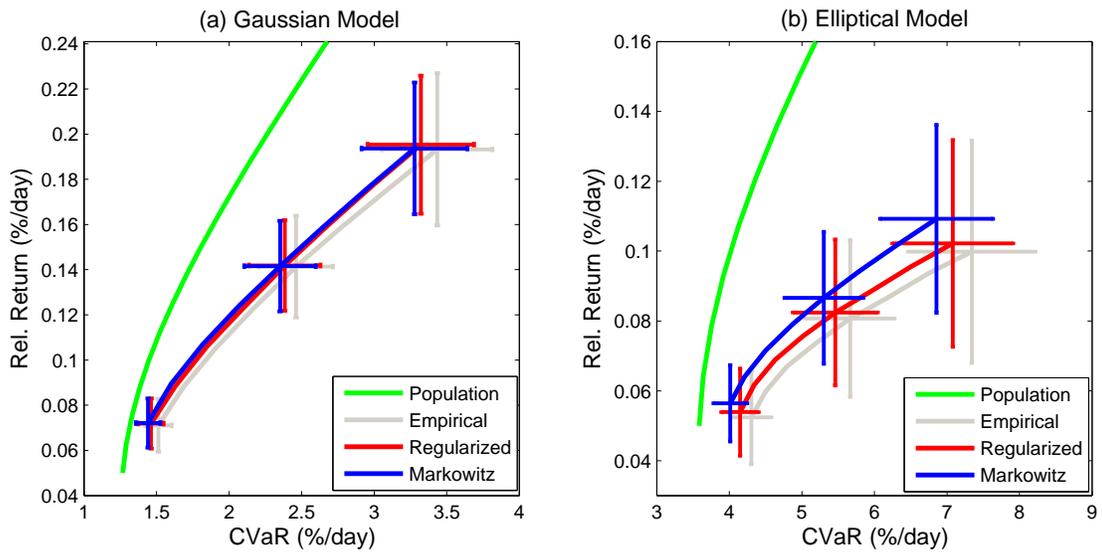}
\caption{\small Average of population risk vs return for solutions to
(\ref{eq:empricial problem}) in grey, (\ref{eq:penalized problem-asympvar}) in red and (\ref{eq:Markowitz}) in blue under
(a) Gaussian model and (b) elliptical model. Green curve denotes the population efficient frontier.
Horizontal and vertical lines show $\pm 1/2$ std error.}\label{fig:gauss_ellip}
\end{figure}

\begin{figure}[t!]
\centering
\includegraphics[width=0.8\paperwidth]{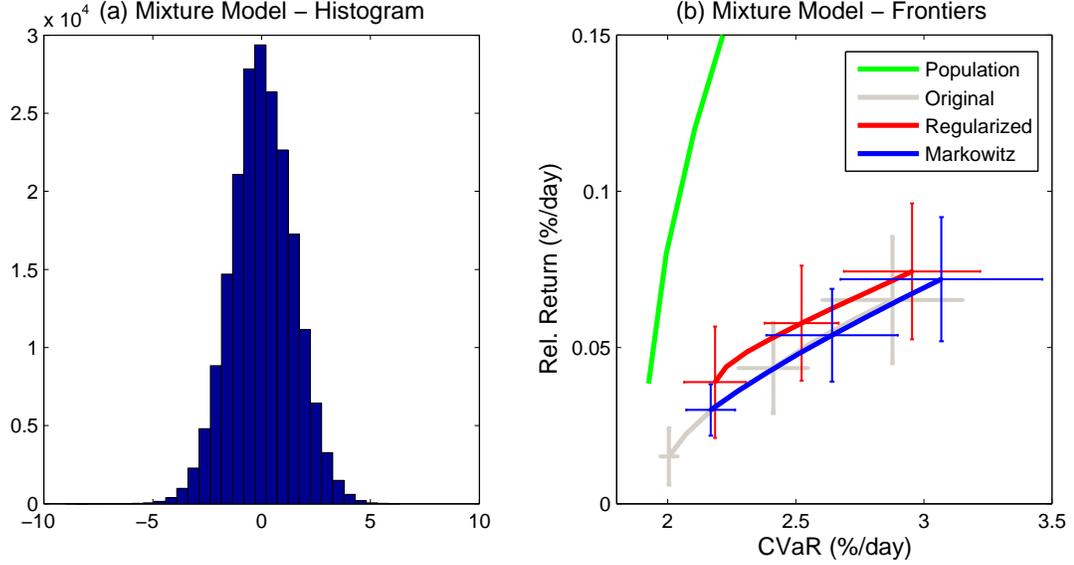}
\caption{\small (a) Distribution of returns for an equally weighted portfolio under the mixture model.
(b) Average of population risk vs return for solutions to (\ref{eq:empricial problem}) in grey,
(\ref{eq:penalized problem-asympvar}) in red and (\ref{eq:Markowitz}) in blue under the mixture model.
Green curve denotes the population efficient frontier. Horizontal and vertical lines show $\pm 1/2$ std error.}\label{fig:mixture}
\end{figure}

\newpage
\appendix
\section{Asymptotics of the CVaR estimator}\label{app:note on cvar}
\textbf{Setting.} Let $\mathsf{L}=[L_{1},\ldots,L_{n}]$ be $n$ iid observations (of portfolio losses) from a distribution $F$ which is absolutely continuous,
has a twice continuously differentiable pdf and a finite second moment.

In this section, we prove the asymptotic distribution of the estimator $\widehat{CVaR}_n(L;\beta)$ introduced in
Eq.~(\ref{eq:CVAR est}) of Sec.~2.1. First, we define a closely related CVaR estimator:
\begin{definition}[\textbf{Type 1 CVaR estimator}.]\label{def:type1 CVaR}
For $\beta\in(0.5,1)$, we define Type 1 CVaR estimator to be
\begin{equation*}
\widehat{CV1}_n(\msf{L};\beta) := \underset{\alpha\in\mathbb{R}}{\min}~~ (1-\varepsilon_n)\alpha+\displaystyle \frac{1}{n-\lceil n\beta \rceil +1}\sum_{i=1}^{n}(L_i-\alpha)^{+},
\end{equation*}
where $\varepsilon_n$ is some constant satisfying $0<\varepsilon_n< (n-\ord+1)^{-1}$, $\sqrt{n}\varepsilon_n\overset{P}{\rightarrow}0$.
\end{definition}

Now consider the following CVaR estimator, expressed without the minimization:
\begin{definition}[\textbf{Type 2 CVaR estimator}.]\label{def:type2 CVaR}
For $\beta\in(0.5,1)$, we define Type 2 CVaR estimator to be
\begin{equation*}
\widehat{CV2}_n(\mathsf{L};\beta) := \frac{1}{n-\lceil n\beta \rceil +1}\sum^{n}_{i=1} L_{i}1(L_{i}\geq \hat{\alpha}_n(\beta)),
\end{equation*}
where $\hat{\alpha}_n(\beta):=L_{(\lceil n\beta\rceil)},$
the $\lceil n\beta\rceil$-th order statistic of the sample $L_1,\ldots,L_n$.
\end{definition}
Type 2 CVaR estimator is asymptotically normally distributed [\citeasnoun{chen2008nee}]. In the remainder of this section,
we show that $\widehat{CV2}_n(\msf{L};\beta)$ is asymptotically equivalent to $\widehat{CV1}_n(\msf{L};\beta)$, which is in turn asymptotically equivalent
to $\widehat{CVaR}_n(\msf{L};\beta)$. We then conclude $\widehat{CVaR}_n(\msf{L};\beta)$ is also asymptotically normal, converging
to the same asymptotic distribution as $\widehat{CV2}_n(\msf{L};\beta)$.

\begin{proposition}\label{prop:alpha min}
The solution $\alpha^*=L_{(\ord)}$ is unique to the one-dimensional optimization problem
\begin{equation*}
\underset{\alpha\in\mathbb{R}}{\min}~~
\left\{G_n(\alpha) := (1-\varepsilon_n)\alpha+\displaystyle \frac{1}{n-\lceil n\beta \rceil +1}\sum_{i=1}^{n}(L_{i}-\alpha)^{+}\right\},
\end{equation*}
where $\varepsilon_n$ is some constant satisfying $0<\varepsilon_n< (n-\ord+1)^{-1}$.
\end{proposition}
\begin{proof}
The expression to be minimized is a piecewise linear convex function with nodes at $L_{1},\ldots,L_{n}$.
We show that $G_n(\alpha)$ has gradients of opposite signs about a single point,
$L_{(\ord)}$, hence this point must be the unique optimal solution. Now consider, for $m\in\{-\ord+1, \ldots, n-\ord\}$:
\begin{eqnarray*}
\Delta(m) &=& G_n(L_{(\ord+m+1)})-G_n(L_{(\ord+m)}) \\
&=& (1-\varepsilon_n)(L_{(\ord+m+1)}-L_{(\ord+m)})-\frac{1}{n-\lceil n\beta \rceil +1}A,
\end{eqnarray*}
where
\begin{eqnarray*}
A
&=& \sum_{i=1}^{n}\left[(L_{i}-L_{(\ord+m+1)})^{+}-(L_{i}-L_{(\ord+m)})^{+}\right] \\
&=& (n-\ord-m)(L_{(\ord+m+1)}-L_{(\ord+m)}).
\end{eqnarray*}
Thus
$$
\Delta(m) = \left(L_{(\ord+m+1)}-L_{(\ord+m)}\right)\left((1-\varepsilon_n)-\frac{n-\ord-m}{n-\lceil n\beta \rceil +1}\right).
$$
Now $\Delta(0)>0$ since $(L_{(\ord+1)}-L_{(\ord)})>0$ and $(1-\varepsilon_n) > (n-\ord)(n-\lceil n\beta \rceil +1)^{-1}$
by the restriction on  $\varepsilon_n$, and $\Delta(-1)<0$ since $(L_{(\ord)}-L_{(\ord-1)})>0$ and $ (1-\varepsilon_n) < 1$
again by the choice of $\varepsilon_n$. Thus $G_n(\alpha)$ has a unique minimum at $\alpha^*=L_{(\ord)}$.
\end{proof}
\textbf{Remark.} Note if $\varepsilon_n=0$, then multiple solutions occur because $\Delta(-1)=0$.
\begin{corollary}\label{cor:equivalence of t1, t2}
Type 1 and Type 2 CVaR estimators are related by
\begin{equation*}
\widehat{CV2}_n(\msf{L};\beta) = \widehat{CV1}_n(\msf{L};\beta)+\varepsilon_nL_{(\ord)}.
\end{equation*}
\end{corollary}
\begin{proof}
Rewriting Type 2 CVaR estimator:
\begin{eqnarray*}
\widehat{CV{2}}_n(\msf{L};\beta)
&=& \frac{1}{n-\lceil n\beta \rceil +1}\sum^{n}_{i=1} L_{i}1(L_{i}\geq L_{(\lceil n\beta \rceil)}) \\
&=& L_{(\lceil n\beta \rceil)}+\frac{1}{n-\lceil n\beta \rceil +1}\sum^{n}_{i=1}(L_{i}-L_{(\lceil n\beta \rceil)})1(L_{i}\geq L_{(\lceil n\beta \rceil)}) \\
&=& \widehat{CV1}_n(\msf{L};\beta)+\varepsilon_nL_{(\ord)},
\end{eqnarray*}
where the final equality is due to Proposition \ref{prop:alpha min}.
\end{proof}

We now show asymptotic normality of $\widehat{CV1}_n(\msf{L};\beta)$.
\begin{lemma}\label{cor:asympt normality}
Type 1 CVaR estimator is asymptotically normal as follows:
\begin{eqnarray}
\frac{\sqrt{n}(1-\beta)}{\gamma_0}\left(\widehat{CV1}_n(\msf{L};\beta) - CVaR(L_1;\beta)\right) \Rightarrow \mathcal{N}(0,1),
\end{eqnarray}
where $\gamma_0^2=Variance[(L_1-\alpha_\beta)1(L_1\geq \alpha_\beta)]$, and $\alpha_\beta=\inf\{\alpha: P(L_1\geq \alpha)\leq 1-\beta\}$, Value-at-Risk of the random loss $L_1$ at level $\beta$.
\end{lemma}
\begin{proof}
Asymptotic normality for Type 2 CVaR estimator is proven in \citeasnoun{chen2008nee},
and the result is immediate from invoking Slutsky's lemma on Corollary \ref{cor:equivalence of t1, t2} and the assumption $\sqrt{n}\varepsilon_n\rightarrow0$.
\end{proof}

\subsection{Proof of Lemma \ref{lem:CV ast}}
The asymptotic distribution of $\widehat{CVaR}_n(\msf{L};\beta)$ is the same as $\widehat{CV1}_n(\msf{L};\beta)$ because
$$\sqrt{n}|\widehat{CVaR}_n(\msf{L};\beta)-\widehat{CV1}_n(\msf{L};\beta)| = o_P(1).$$

\section{Proof of Theorem 1}\label{app: qcqp proof}
\begin{lemma}\label{lemma:1}
Consider the optimization problem
\begin{eqnarray}\label{eq:general relax}
     \begin{array}{crll}
      \underset{{z}\in\mbb{R}^n}{\min} & {z}\trsp 1_n &&\\
      s.t. &  z_i \geq & 0& \forall~i \\
            &  z_i \geq & c_i & \forall~i \\
         & {z}\trsp \Omega_n{z} \leq& f&
\end{array}
\end{eqnarray}
where $c_i>0~\forall~i$, $f>0$, $\Omega_n =(n-1)^{-1}(I_n-n^{-1}{1}_n{1}_n\trsp )$, the sample covariance operator. Suppose (\ref{eq:general relax})
is feasible with an optimal solution $({x}^*,{z}^*)$. Let $S_1({z}) : = \{1\leq i \leq n:~ z_i=0 \}$,
$S_2({z}) : = \{1\leq i \leq n:~ z_i=c_i\}$ and $V({z}):=S_1^c\cap S_2^c$ (i.e.~$V({z})$ is the set of indices for which $z_i> \max(0,c_i)$). Then the optimal solution ${z}^*$ falls into one of two cases: either $S_1({z}^*)\neq \emptyset$ and $V({z}^*)=\emptyset$, or
$S_1({z}^*) = \emptyset$ and $V({z}^*)\neq\emptyset$.
\end{lemma}

\begin{proof}
The problem (\ref{eq:general relax}) is a convex optimization problem because $\Omega_n$ is a positive semidefinite matrix. The problem is also strictly feasible, since ${z}_0 = 2\max_i\{c_i\}1_n$ is a strictly feasible point: clearly, $z_{0,i} > \max\{0,c_i\} ~\forall~i$ and ${z}_0\trsp \Omega_n{z}_0 = 0 < f$ as $1_n$ is orthogonal to $\Omega_n$. Thus Slater's condition for strong duality holds, and we can derive properties of the optimal solution by examining KKT conditions.

The Lagrangian is
\begin{eqnarray*}
&&\mathcal{L}({z},{\eta}_1,{\eta}_2,\lambda)= \lambda{z}\trsp \Omega_n{z}+(1_n-{\eta}_1-{\eta}_2)\trsp {z} +{\eta}_2\trsp {c}-\lambda f
\end{eqnarray*}
The KKT conditions are
\begin{itemize}
\item Primal feasibility
\item Dual feasibility: ${\eta}^*_1,{\eta}^*_2\geq0$  component-wise and $\lambda^* \geq0$
\item Complementary slackness:

$z^*_i\eta^*_{1,i} = 0~\forall~i$, $(z^*_i-c_i)\eta^*_{2,i} = 0~\forall~i$ and $\lambda^*[({z}^*)\trsp \Omega_n{z}^*- f]=0$
\item First Order Condition:
\begin{subequations}\label{eq:FOC}
\begin{eqnarray}
\nabla_{{z}^*}\mathcal{L} &=& 2\lambda\Omega_n{z}^*+(1_n-{\eta}^*_1-{\eta}^*_2)= {0} \label{subeq:foc z}
\end{eqnarray}
\end{subequations}
\end{itemize}
By substituting for $\Omega_n$, (\ref{subeq:foc z}) can be written as
\begin{eqnarray}\label{subeq: foc z 2}
\frac{2\lambda}{n-1}\left({z}^*-\frac{1}{n}(1_n\trsp {z}^*)1_n\right) &=& -1_n+{\eta}^*_1+{\eta}^*_{2}.
\end{eqnarray}

Suppose $S_1({z}^*)\neq\emptyset$ at the optimal primal-dual point $({z}^*,{\eta}^*_1,{\eta}^*_2,\lambda^*)$.
Then $\exists~i_0\in S_1({z}^*)$ such that $z_{i_0}
^*=0$. The $i_0$-th component of (\ref{subeq: foc z 2}) gives
\begin{eqnarray}\label{eq:S1 cond}
\displaystyle-\frac{2\lambda^*}{n(n-1)}(1_n\trsp {z}^*) &=& -1+\eta^*_{1,i_0}+\eta^*_{2,i_0}.
\end{eqnarray}

Now suppose $V({z}^*)\neq\emptyset$ at the optimal primal-dual point $({z}^*,{\eta}^*_1,{\eta}^*_2,\lambda^*)$.
Then  $\exists~j_0\in V({z}^*)$ such that $z_{j_0}^*>\max(0,c_i)$, $\eta^*_{1,j_0}=0$ and $\eta^*_{2,j_0}=0$. The $j_0$-th component of (\ref{subeq: foc z 2}) gives
\begin{eqnarray}\label{eq:V cond}
\frac{2\lambda^*}{n-1}\left(z_{j_0}^*-\frac{1}{n}(1_n\trsp {z}^*)\right) &=& -1,
\end{eqnarray}
which also implies $\lambda^*>0$.

Now suppose $S_1({z}^*)$ and $V({z}^*)$ are both nonempty. Combining (\ref{eq:S1 cond}) and (\ref{eq:V cond}), we arrive at
the necessary condition
\begin{equation*}
\frac{2\lambda^*}{n-1} z_{j_0}^* = -\eta^*_{1,i_0}-\eta^*_{2,i_0}.
\end{equation*}
which is clearly a contradiction since $LHS>0$ whereas $RHS\leq0$. Hence $S_1({z}^*)$ and $V({z}^*)$ cannot both be nonempty.
\end{proof}

\subsection{Proof of Theorem 1}
\begin{proof}
Clearly, (\ref{eq:penalized problem-QCQP}) is a relaxation of (\ref{eq:penalized problem-asympvar}): the components of the variable ${z}$ in (\ref{eq:penalized problem-QCQP}) are relaxations of $\max(0,-w\trsp X_i-\alpha)$. Thus the two problem formulations are equivalent if at optimum, $z_i = \max(0,-w\trsp X_i-\alpha)~\forall~i=1,\ldots,n$ for (\ref{eq:penalized problem-QCQP}).
\\
\\
Let $(\alpha^*,w^*,{z}^*,\nu_1^*,\nu_2^*,{\eta}^*_1,{\eta}^*_2,\lambda^*_1,\lambda^*_2)$ be the primal-dual optimal point for (\ref{eq:penalized problem-QCQP}) and (\ref{eq:dual of reg}). Our aim is to show that $V({z}^*)$, the set of indices for which $z^*_i>\max(0,-w\trsp X_i-\alpha)$, is empty. Suppose the contrary. Then by Lemma \ref{lemma:1}, $S_1({z}^*)$, the set of indices for which $z^*_i=0$, is empty. This means $z^*_i>0~\forall~i$ and $\eta^*_{1,i}=0~\forall~i$ by complementary slackness.

Now consider the sub-problem for a fixed ${\eta}_2$ in the dual problem (\ref{eq:dual of reg}):
\begin{eqnarray}\label{eq:eta1}
\underset{{\eta}_1:{\eta}_1\geq{0}}{\max} -({\eta}_1+{\eta}_2)\Omega^\dagger_n ({\eta}_1+{\eta}_2).
\end{eqnarray}
As $1_n$ is orthogonal to $\Omega_n^\dagger$, and $\Omega_n^\dagger$ is positive semidefinite, the optimal solution is of the form
${\eta}_1 = a1_n-{\eta}_2$, where $a$ is any constant such that $a\geq\max_i(\eta_{2,i})$, with a corresponding optimal objective $0$. Hence, bearing in mind the constraints ${\eta}_2\geq{0}$ and ${\eta}_2\trsp 1_n=1$ in (\ref{eq:dual of reg}), ${\eta}_1 = {0}$ is one of the optimal solutions iff ${\eta}_2^*= 1_n/n$. Thus if ${\eta}_2^*\neq 1_n/n$, we get a contradiction. Otherwise, we can force the dual problem to find a solution with ${\eta}_1\neq{0}$ by introducing an additional constraint ${\eta}_1\trsp 1_n\geq \delta$ for some constant $0<\delta \ll 1$.
\end{proof}

\section{Details of asymptotic theory}

\subsection{Proof of Theorem \ref{theo:consistency unreg}}\label{App:const proof}

\begin{proof}
By uniqueness of $\theta_0(\lambda_1,\lambda_2)$ and Assumption 1 (and compactness arguments), for every $\varepsilon>0$, there exists $\eta>0$ such that
$$||\hat{\theta}_n(\lambda_1,\lambda_2) - \theta_0(\lambda_1,\lambda_2)||_2>\varepsilon \implies M(\hat{\theta}_n;\lambda_1,\lambda_2)-M(\theta_0;\lambda_1,\lambda_2)>\eta.$$
Thus if we can show the probability of the event $\{M(\hat{\theta}_n;\lambda_1,\lambda_2)-M(\theta_0;\lambda_1,\lambda_2)>\eta\}$ goes to zero for every $\varepsilon>0$, then
we have consistency.

We also have
\begin{equation}\tag{$\star$}\label{eq:inter}
M_n(\hat{\theta}_n;\lambda_1,\lambda_2)\leq M_n(\theta_0;\lambda_1,\lambda_2)+o_P(1) = M(\theta_0;\lambda_1,\lambda_2)+o_P(1),
\end{equation}
the first inequality because $\hat{\theta}_n(\lambda_1,\lambda_2)$ is a near-minimizer of $M_n$, and the second equality by the Weak Law of Large Numbers (WLLN) on $M_n(\theta_0;\lambda_1,\lambda_2)$.

Thus
\begin{eqnarray*}
0&\leq& M(\hat{\theta}_n;\lambda_1,\lambda_2)-M(\theta_0;\lambda_1,\lambda_2)\\
&=& [M(\hat{\theta}_n;\lambda_1,\lambda_2)-M_n(\hat{\theta}_n;\lambda_1,\lambda_2)]+[M_n(\hat{\theta}_n;\lambda_1,\lambda_2)-M_n(\theta_0;\lambda_1,\lambda_2)]+[M_n(\theta_0;\lambda_1,\lambda_2)-M(\theta_0;\lambda_1,\lambda_2)] \\
&\leq& M(\hat{\theta}_n;\lambda_1,\lambda_2)-M_n(\hat{\theta}_n;\lambda_1,\lambda_2)+o_P(1),
\end{eqnarray*}
because the second term in $[~]$ is $o_P(1)$ by ($\star$), and the last term in $[~]$ is $o_P(1)$ by WLLN.
We are left to prove $|M_n(\hat{\theta}_n;\lambda_1,\lambda_2)-M(\hat{\theta}_n;\lambda_1,\lambda_2)|\overset{P}{\rightarrow}0$. At first glance, one may consider invoking the WLLN again.
However, as $\hat{\theta}_n(\lambda_1,\lambda_2)$ is a random sequence of vectors that changes for every $n$, we cannot apply the WLLN which is a pointwise result (i.e.~for each fixed $\theta\in\Theta$), and we need to appeal to the stronger ULLN.

\textbf{Case I:} $\lambda_1=\lambda_2=0$.
To show ULLN for the original objective, we show that $\mathcal{F}_1$ is a Lipschitz class of functions, hence $N_{[~]}(\varepsilon,\mathcal{F}_1,L_r(P))$ for every $\varepsilon>0$. Now  $\theta\mapsto m_\theta(x) = \alpha + (1-\beta)^{-1}(-\alpha-w_0^{\trsp}x-v\trsp L^{\trsp}x)^{+}$ is clearly differentiable at $\theta_0$ for all $x\in\mbb{R}^p$. Furthermore,
\begin{eqnarray*}
\nabla_\theta m_{\theta}(x) =
\left[\begin{array}{c}
  -1 \\
  -L\trsp x
\end{array}\right]I(x),
\end{eqnarray*}
where $I(x):=\mbb{I}(-\alpha-w_0^{\trsp}x-v\trsp L^{\trsp}x\geq 0)$, hence
\begin{equation}\label{eq:dotm}
\dot{m}(x) := \max(1,||L\trsp x||_{\infty})
\end{equation}
is an upper bound on $||\nabla_\theta m_{\theta}(x)||_{\infty}$ and is independent of $\theta$. Thus
$|m_{\theta_1}(x)-m_{\theta_2}(x)| \leq  \dot{m}(x)||\theta_1-\theta_2||_2$ for all $\theta_1,\theta_2\in[-K,K]^{1+p}$, and together with Assumption 2 (here a weaker assumption that $X$ has finite second moment suffices), $\mathcal{F}_1$ is a Lipschitz class.

\textbf{Case II:} $\lambda_1\geq0,\lambda_2\geq0$, $\lambda_1,\lambda_2$ not both zero.
Corollary 3.5 in \citeasnoun{arcones1993limit} says that ULLN also holds for the penalized objective if $N_{[~]}(\varepsilon,\mathcal{F}_2,L_2(P\times P))<\infty$ for every $\varepsilon>0$. Let us now show that $\mathcal{F}_2$ is also a Lipschitz class of functions.
Again, it is clear that
$$\theta\mapsto m^U_{(\theta;\lambda_1,\lambda_2)}(x_1,x_2) = \frac{1}{2}\left[m_{\theta}(x_1)+m_{\theta}(x_2)\right]+ \frac{\lambda_1}{2}[(w_1+Lv)\trsp (x_1-x_2)]^2 +\frac{\lambda_2}{2}(z_\theta(x_1)-z_\theta(x_2))^2$$
is differentiable at $\theta_0$ for all $(x_1,x_2)\in\mbb{R}^p\times \mbb{R}^p$.
Also for all $\theta\in[-K,K]^{1+p}$,
\begin{eqnarray*}
\nabla_\theta \frac{\lambda_1}{2}[(w_1+Lv)\trsp (x_1-x_2)]^2
&=&
\lambda_1 (x_1-x_2)(x_1-x_2)\trsp(w_1+Lv) \\
\implies
||\nabla_\theta \frac{\lambda_1}{2}[(w_1+Lv)\trsp (x_1-x_2)]^2||_\infty
&\leq& \lambda_1 ||x_1-x_2||_\infty^2||w_1+Lv||_\infty
\leq \lambda_1 C(K)||x_1-x_2||_\infty^2\\
&&~~~~~\text{for some constant}~C(K)~\text{dependent on}~K,~\text{and} \\
\nabla_\theta \frac{\lambda_2}{2}(z_\theta(x_1)-z_\theta(x_2))^2
&=&
\lambda_2(z_\theta(x_1)-z_\theta(x_2))
\left[\begin{array}{c}
  -I(x_1)+I(x_2) \\
 -L\trsp x_1I(x_1)+L\trsp x_2I(x_2)
\end{array}\right],~\text{and}\\
|z_\theta(x_1)|
&=&|-(\alpha-w_0\trsp x_1-v\trsp L\trsp x_1)^+| \\
&\leq& |\alpha-w_0\trsp x_1-v\trsp L\trsp x_1| \leq K+|w_0\trsp x_1|+K|e\trsp x_1|
\\
\implies
||\nabla_\theta \frac{\lambda_2}{2}(z_\theta(x_1)-z_\theta(x_2))^2||_\infty
&\leq&
\lambda_2|z_\theta(x_1)-z_\theta(x_2)|(\dot{m}(x_1)+\dot{m}(x_2)) \\
&&~~~~~~~~~~~~~~~~~~~~~~\text{$\dot{m}$ as defined in Eq.~(\ref{eq:dotm})} \\
&\leq& \lambda_2C'(K)(||x_1||_\infty+||x_2||_\infty)(\dot{m}(x_1)+\dot{m}(x_2)),\\
&&~~~~~\text{for some constant}~C'(K)~\text{dependent on}~K,
\end{eqnarray*}
hence
\begin{equation}\label{eq:dotm^U}
\dot{m}^U_{(\lambda_1,\lambda_2)}(x_1,x_2) :=
\frac{1}{2}[\dot{m}(x_1)+\dot{m}(x_2)]+\lambda_1 C(K)||x_1-x_2||_\infty^2+\lambda_2C'(K)(||x_1||_\infty+||x_2||_\infty)(\dot{m}(x_1)+\dot{m}(x_2))
\end{equation}
is an upper bound on $||\nabla_\theta m^U_{(\theta;\lambda_1,\lambda_2)}(x_1,x_2)||_{\infty}$ that is independent of $\theta$.
Thus
$$|m^U_{(\theta_1;\lambda_1,\lambda_2)}(x_1,x_2)-m^U_{(\theta_2;\lambda_1,\lambda_2)}(x_1,x_2)| \leq \dot{m}^U_{(\lambda_1,\lambda_2)}(x_1,x_2)||\theta_1-\theta_2||_2,$$
and together with Assumption 2, $\mathcal{F}_2$ is a Lipschitz class.
\end{proof}

\subsection{Proof of Theorem \ref{theo:Uasympt}}\label{App:CLT proof}
In what follows, we suppress the dependence on $\lambda_1,\lambda_2$ for notational convenience.

\begin{proof}
The proof parallels the proof of Theorem 5.23 of \citeasnoun{vandervaart2000}. Let us assume for now that
\begin{enumerate}
\item For every given random sequence ${h}_n$ that is bounded in probability,
\begin{equation*}\tag{*}\label{eq:*}
    \mbb{U}_{n}[\sqrt{n}(m^U_{\theta_0+{h}_n/\sqrt{n}}-m^U_{\theta_0})-{h}_n\trsp \dot{m}^U_{\theta_0}]\overset{P}{\rightarrow}0,
\end{equation*}
and
\item $\sqrt{n}(\hat{\theta}_n-\theta_0) = O_{P}(1)$.
\end{enumerate}

Since $\theta\mapsto M(\theta)$ is twice-differentiable, and $\nabla_\theta M(\theta)|_{\theta=\theta_0} =0$ by first-order condition, we can rewrite Eq.~(\ref{eq:*}) to get
\begin{eqnarray*}
n{n\choose 2}^{-1}\sum_{i\neq j}[m^U_{\theta_0+{h}_n/\sqrt{n}}(X_i,X_j)-m^U_{\theta_0}(X_i,X_j)]
&=& \frac{1}{2}{h}_n\trsp V_{\theta_0}{h}_n + {h}_n\trsp \mbb{U}_n[\dot{m}^U_{\theta_0}] +o_p(1) \\
&=& \frac{1}{2}{h}_n\trsp V_{\theta_0}{h}_n + {h}_n\trsp \mbb{G}_n[\dot{m}^1_{\theta_0}] + o_p(1),
\end{eqnarray*}
where we use the fact, from Hoeffding decomposition,
\begin{eqnarray*}
\mbb{U}_n[\dot{m}^U_{\theta_0}] &=& \frac{\sqrt{n}}{{n\choose 2}}\sum_{i\neq j}
\left[\dot{m}^U_{\theta_0}(X_i,X_j)-\mbb{E}_{X_1,X_2}[\dot{m}^U_{\theta_0}(X_1,X_2)]\right]\\
&=&\frac{1}{\sqrt{n}}\sum_{i=1}^n[\dot{m}^1_{\theta_0}(X_i)-\mbb{E}\dot{m}_{\theta_0}^1(X_1)]+o_{p}(1)
= \mbb{G}_n[\dot{m}^1_{\theta_0}]+o_{p}(1),
\end{eqnarray*}
with $\dot{m}^1_{\theta}$ as in the statement of the theorem.

The above statement is valid for both $\hat{h}_n=\sqrt{n}(\hat{\theta}_n-\theta_0)$ and for $\tilde{h}_n=-V_{\theta_0}^{-1}\mbb{G}_n\dot{m}^1_{\theta_0}$.
Upon substitution, we obtain
\begin{eqnarray*}
n{n\choose 2}^{-1}\sum_{i\neq j}[m^U_{\theta_0+\hat{h}_n/\sqrt{n}}(X_i,X_j)-m^U_{\theta_0}(X_i,X_j)]
&=& \frac{1}{2}\hat{h}_n\trsp V_{\theta_0}\hat{h}_n + \hat{h}_n\trsp \mbb{G}_n[\dot{m}^1_{\theta_0}] + o_p(1) \\
\leq n{n\choose 2}^{-1}\sum_{i\neq j}[m^U_{\theta_0+\tilde{h}_n/\sqrt{n}}(X_i,X_j)-m^U_{\theta_0}(X_i,X_j)]
&=& -\frac{1}{2}\mbb{G}_n[\dot{m}^1_{\theta_0}]\trsp V_{\theta_0}^{-1}\mbb{G}_n[\dot{m}^1_{\theta_0}]+o_{p}(1)
\end{eqnarray*}
where the inequality is from the definition of $\hat{\theta}_n = \theta_0+\hat{h}_n/\sqrt{n}$ as a near-minimizer.

Taking the difference and completing the square, we get
\begin{eqnarray*}
\frac{1}{2}(\hat{h}_n+V_{\theta_0}^{-1}\mbb{G}_n\dot{m}^1_{\theta_0})\trsp V_{\theta_0}(\hat{h}_n+V_{\theta_0}^{-1}\mbb{G}_n\dot{m}^1_{\theta_0})
+o_{p}(1)\leq0,
\end{eqnarray*}
and because $V_{\theta_0}$ is nonsingular, the quadratic form on the left must converge to zero in probability.
The same must be true for $||\hat{h}_n+V_{\theta_0}^{-1}\mbb{G}_n\dot{m}^1_{\theta_0}||_2$.

To complete the proof, we need to show (\ref{eq:*}) and $\sqrt{n}(\hat{\theta}_n-\theta_0)=O_P(1)$ hold.

\vspace{1cm}
\textbf{Proof of (\ref{eq:*}).}
\\Let $f_h := \sqrt{n}(m^U_{\theta_0+h/\sqrt{n}}-m^U_{\theta_0})-h\trsp \dot{m}^U_{\theta_0}$. As we are considering only sequences $h_n$
that are bounded in probability, it suffices to show $\sup_{h:||h||_2\leq 1}|\mbb{U}_{n}[f_h]|$ goes to zero in probability.
Again by Hoeffding decomposition, for any given random sequence $h_n$ that is bounded in probability,
$\mbb{U}_{n}[f_{h_n}]=\mbb{G}_{n}[f^1_{h_n}]+E_n({h}_n)$, where $f^1_{h}$ is the first term in the Hoeffding decomposition of $\mbb{U}_{n}[f_{h}]$ given by
\begin{eqnarray*}
f^1_{h} &=& \sqrt{n}(m^1_{\theta_0+{h}/\sqrt{n}}-m^1_{\theta_0})-{h}\trsp \dot{m}^1_{\theta_0}, \\
{m}^1_{\theta}(x_1) &=& 2\mbb{E}_{X_2}[{m}^U_{\theta}(x_1,X_2)]-\mbb{E}_{X_1,X_2}[{m}^U_{\theta}(X_1,X_2)],
\end{eqnarray*}
and $\dot{m}^1_{\theta}$ as defined in the statement of the theorem.
According to Lemma 19.31 in \citeasnoun{vandervaart2000}, if $\mathcal{F}_2':=\{m^1_\theta:\theta\in[-K,K]^{1+p}\}$ is a Lipschitz class of functions,
$$\underset{h:||h||_2\leq 1}{\sup}|\mbb{G}_{n}[f_{h_n}^1]| \overset{P}{\rightarrow}0.$$
Now by Assumption 2 that $X_i$'s are iid continuous random vectors with finite fourth moment, $\theta\mapsto m^1_\theta(x)$ is differentiable
at $\theta_0$ for all $x\in\mbb{R}$. Further, by triangle inequality,
\begin{eqnarray*}
|m^1_{\theta_1}(x)-m^1_{\theta_2}(x)| &\leq&
 2\mbb{E}_{X_2}|{m}^U_{\theta_1}(x,X_2)-{m}^U_{\theta_2}(x,X_2)|+\mbb{E}_{X_1,X_2}|{m}^U_{\theta}(X_1,X_2)-{m}^U_{\theta}(X_1,X_2)| \\
 &\leq & m^1(x)||\theta_1-\theta_2||_2,
\end{eqnarray*}
where $m^1(x)=(2\mbb{E}_{X_2}|\dot{m}^U(x,X_2)|+\mbb{E}_{X_1,X_2}|\dot{m}^U(X_1,X_2)|)$, $\dot{m}^U$ as in Eq.~(\ref{eq:dotm^U}).
Since $X_i$'s have finite fourth moment, $\mbb{E}[m^1(X_1)^2]<\infty$ and thus $\mathcal{F}_2'$ is a Lipschitz class.

Now we are left to show ${\sup}_{h:||h||_2\leq 1}|E_n(h)|\overset{P}{\rightarrow}0$. Let $\mathcal{F}_h:=\{f_h:||h||_2\leq1\}$.
According to Theorem 4.6 of \citeasnoun{arcones1993limit}, ${\sup}_{h:||h||_2\leq 1}|E_n(h)|\overset{P}{\rightarrow}0$ if
$\mathcal{F}_h$ has a finite, integrable envelope function and both $\mathcal{F}_h$
and $\mathcal{F}_h':=\{f^1_h:||h||_2\leq1\}$ are Lipschitz classes
about $h=0$. $\mathcal{F}_h$ has a finite, integrable envelope function  $F(x_1,x_2) = \dot{m}^U(x_1,x_2)+||\dot{m}_{\theta_0}(x_1,x_2)||_2<\infty$
due to Assumption 2 and the Lipschitz property of $m^U_\theta$:
\begin{eqnarray*}
|f_h| &\leq& |\sqrt{n}(m^U_{\theta_0+h/\sqrt{n}}-m^U_{\theta_0})-h\trsp \dot{m}_{\theta_0}|\\
&\leq& (\dot{m}^U+||\dot{m}_{\theta_0}||_2)||h||_2.
\end{eqnarray*}
It is now straight-forward to check that $\mathcal{F}_h$ is a Lipschitz class about $h=0$, and $\mathcal{F}_h'$ also, because it inherits the key
properties from $\mathcal{F}_h$.

\vspace{1cm}
\textbf{Proof of $\sqrt{n}(\hat{\theta}_n-\theta_0)=O_P(1)$}.
\\The proof of $\sqrt{n}(\hat{\theta}_n(0,0)-\theta_0(0,0)) = O_{p}(1)$ can be found in Theorem 5.52 and Corollary 5.53 of \citeasnoun{vandervaart2000}, and is a standard M-estimation result. In essence, Theorem 5.52 shows that, under some regularity conditions,
$P(\sqrt{n}||\hat{\theta}_n(0,0)-\theta_0(0,0)||_2> \alpha)$ can be bounded by $P(|\mbb{G}_n[m_\theta]|> \alpha') = P(\sqrt{n}|M_n(\theta)-M(\theta)|> \alpha')$, which is shown to go to zero via some maximal inequalities. Corollary 5.53 shows that the Lipschitz condition on $\{m_\theta:\theta\in[-K,K]^{1+p}\}$ is sufficient to satisfy the regularity conditions of the theorem.

We can extend Theorem 5.52 to show $\sqrt{n}(\hat{\theta}_n(\lambda_1,\lambda_2)-\theta_0(\lambda_1,\lambda_2))$, $\lambda_1,\lambda_2\geq 0$
not both zero, by bounding $P(\sqrt{n}||\hat{\theta}_n(\lambda_1,\lambda_2)-\theta_0(\lambda_1,\lambda_2)||_2> \alpha)$ by
$$P(|\mbb{U}_n[m^U_\theta]|> \alpha')\leq P(|\mbb{G}_n[m^1_\theta]|+|E'_n(\theta)|> \alpha'),$$
where $E_n'$ is the remainder term after first-order projection of the U-process $\mbb{U}_n[m^U_\theta]$.
It remains to show that for every sufficiently small $\delta>0$,
\begin{equation}
\underset{\theta:||\theta-\theta_0||_2< \delta}{\sup}~|E'_n(\theta)|\overset{P}{\rightarrow}0,
\end{equation}
which can be proven using the same reasoning for $\underset{h:||h||_2\leq 1}{\sup}|E_n(h)|\overset{P}{\rightarrow}0$ in the proof of (\ref{eq:*}).
\end{proof}

\subsection{Computation of key statistics}\label{App:key statistics}

Given the distribution for $X$, both $A_0=A_{\theta_0}(0,0)$ and $B_0=B_{\theta_0}(0,0)$ are computable. The lemma below computes the key quantities that constitute $A_0$ and $B_0$ when $X\sim\mathcal{N}(\mu,\Sigma)$.
\begin{lemma}\label{lemma:key stats}
Suppose $X\sim\mathcal{N}(\mu,\Sigma)$, and $z_{\theta}(X) = -\alpha-w\trsp X\sim \mathcal{N}(\mu_1,\sigma_1^2)$, where
$\mu_1 = -\sigma_1\Phi^{-1}(\beta)$ and $\sigma_1^2 = w\trsp \Sigma w$.
Then
\begin{eqnarray}
p_0&=& P(z_{\theta}(X)=0) = \frac{1}{\sqrt{2\pi}\sigma_1}\exp{\left(-\frac{1}{2}\Phi^{-1}(\beta)^2\right)}\label{subeq:p0} \\
\Exp[\max(z_{\theta}(X),0)]
&=& \frac{\sigma_1}{\sqrt{2\pi}}\exp\left(-\frac{1}{2}\Phi^{-1}(\beta)^2\right)-\sigma_1(1-\beta)\Phi^{-1}(\beta)\label{subeq:Ezp}\\
\Exp[L_j\trsp X\mbb{I}(Z_{1}\geq0)]
&=&  (1-\beta)(L_j\trsp \mu-\Phi^{-1}(\beta)\frac{L_j\trsp \Sigma w}{\sigma_{1}})-\frac{L_j\trsp \Sigma w}{\sigma_{1}^2}\Exp[\max(z_{\theta}(X),0)]\label{subeq:ELLZI}\\
\Exp[L_j\trsp X|Z_{1}=0]
&=&  L_j\trsp \mu-\Phi^{-1}(\beta)\frac{L_j\trsp  \Sigma w}{\sigma_1}\label{subeq:1}\\
\Exp[L_j\trsp XL_l\trsp X\mbb{I}(Z_{1}\geq 0)]
&=&\frac{1}{4}(g(\mu_1,(L_j+L_l)\trsp \mu,\sigma_1,\sigma_2,-(L_j+L_l)\trsp \Sigma w_{1})\nonumber\\
&&~~~~~~~~~-g(\mu_1,(L_j-L_l)\trsp \mu,\sigma_1,\sigma_2,-(L_j-L_l)\trsp \Sigma w_{1}))\label{subeq:2}\\
\Exp[L_j\trsp XL_l\trsp X|Z_{1}= 0]
&=&\frac{1}{4}(h(\mu_1,(L_j+L_l)\trsp \mu,\sigma_1,\sigma_2,-(L_j+L_l)\trsp \Sigma w_{1})\nonumber\\
&&~~~~~~~~~-h(\mu_1,(L_j-L_l)\trsp \mu,\sigma_1,\sigma_2,-(L_j-L_l)\trsp \Sigma w_{1}))\label{subeq:3}
\end{eqnarray}
where
\begin{eqnarray*}
g(\mu_1,\mu_2,\sigma_1,\sigma_2,\sigma_{12})
&=&(1-\beta)\left[\mu_2^2+\sigma_{2}^2\right]+p_0\sigma_{12}\left[-\Phi^{-1}(\beta)\frac{\sigma_{12}}{\sigma_{1}} +2\mu_2\right]\\
h(\mu_1,\mu_2,\sigma_1,\sigma_2,\sigma_{12})
&=&=(\mu_2+\frac{\sigma_{12}}{\sigma_{1}}\Phi^{-1}(\beta))^2+\sigma^2_{2}
-\frac{\sigma_{12}^2}{\sigma_{1}^2}\;.
\end{eqnarray*}
\end{lemma}

\begin{proof}
We use the fact that if $Z_1\sim\mathcal{N}(\mu_1,\sigma_1)$ and $Z_2\sim\mathcal{N}(\mu_2,\sigma_2)$,
\begin{equation}\label{eq:condGauss}
Z_2|Z_1={\cal N} (\mu_2+\sigma_{12}/\sigma^2_{1} (Z_1-\mu_1),\sigma^2_{2}-\sigma_{12}^2/\sigma^2_{1}),
\end{equation}
where $\sigma_{12}=Cov(Z_1,Z_2)$.

$\bullet$ \textbf{Terms involving only $L_j\trsp X$.}\\
Note that from (\ref{eq:condGauss}), $\mbb{E}[Z_2|Z_1=0]=\mu_2-\frac{\sigma_{12}}{\sigma_{1}^2}\mu_1$.
Let $Z_2=L_j\trsp X$, and recall that $\Exp(L_j\trsp X)=L_j\trsp \mu$ and $\Exp(Z_1)=-\sigma_1\Phi^{-1}(\beta)$.
Also, note that $\sigma_{12}=-L_j\trsp \Sigma w$. After some algebra, we get (\ref{subeq:1}).

Since we know the distribution of $Z_2|Z_1$, we have
\begin{eqnarray*}
\Exp[Z_2\mbb{I}(Z_1\geq 0)]
&=&\Exp[\mbb{I}(Z_1\geq 0)(\mu_2+\frac{\sigma_{12}}{\sigma_{1}^2}(Z_1-\mu_1))]\\
&=&(1-\beta)(\mu_2-\frac{\sigma_{12}}{\sigma_{1}^2}\mu_1)+\frac{\sigma_{12}}{\sigma_{1}^2}\Exp[Z_1\mbb{I}(Z_1\geq 0)]\\
&=&(1-\beta)(L_j\trsp \mu-\Phi^{-1}(\beta)\frac{L_j\trsp \Sigma w}{\sigma_1})-\frac{L_j\trsp \Sigma w}{\sigma_{1}^2}\Exp[\max(Z_1,0)]
\end{eqnarray*}

$\bullet$ \textbf{Terms involving $L_j\trsp XL_l\trsp X$.}\\
To compute $\mbb{E}[L_j\trsp XL_l\trsp X\mbb{I}(Z_{1}\geq 0)]$ and $\mbb{E}[L_j\trsp XL_l\trsp X|Z_{1}=0]$,
first note that
$$
\Exp[L_j\trsp XL_l\trsp X\mbb{I}(Z_{1}\geq 0)]=\frac{1}{4}\Exp\left[[(L_j\trsp X+L_l\trsp X)^2-(L_j\trsp X-L_l\trsp X)^2] \mbb{I}(Z_{1}\geq 0)\right]\;.
$$
and similarly
$$
\Exp[L_j\trsp XL_l\trsp X|Z_{1} = 0]=\frac{1}{4}\Exp[\left[(L_j\trsp X+L_l\trsp X)^2-(L_j\trsp X-L_l\trsp X)^2\right]|Z_{1}= 0]\;.
$$
Hence it is sufficient to first find expressions for $\mbb{E}[Z_2^2 \mbb{I}(Z_1\geq 0)]$ and $\mbb{E}[Z_2^2 |Z_1= 0]$ for some normal $Z_2$, then
apply the resulting formulae to $Z_2 = (L_j\pm L_l)\trsp X$. This results in $\mu_2=(L_j\pm L_l)\trsp \mu$, $\sigma_{12}=-(L_j\pm L_l)\trsp \Sigma w $ and
$\sigma_{2}^2=(L_j\pm L_l)\trsp \Sigma (L_j\pm L_l)$.

From tower property and the conditional distribution of $Z_2|Z_1$,
$$
\Exp[Z_2^2\mbb{I}(Z_1\geq 0)]
=\Exp[\mbb{I}(Z_1\geq 0)\left[(\mu_2+\frac{\sigma_{12}}{\sigma_{1}^2}(Z_1-\mu_1))^2+\sigma^2_{2}-\frac{\sigma_{12}^2}{\sigma^2_{1}}\right]]\;.
$$
By simple computations,
\begin{align*}
\Exp[(Z_1-\mu_1)\mbb{I}_{Z_1\geq 0}]&=\frac{\sigma_{1}}{\sqrt{2\pi}}\exp(-\mu_1^2/(2\sigma_{1}^2))=\sigma_1^2 f_{Z_1}(0)=\sigma_1^2 p_0\;, \text{ and} \\
\Exp[(Z_1-\mu_1)^2\mbb{I}_{Z_1\geq 0}]&= \sigma_{1}^2 (\mu_1 p_0 +(1-\beta))\;.
\end{align*}
Now $\mu_1/\sigma_{1}=-\Phi^{-1}(\beta)$, and
\begin{align*}
\Exp[Z_2^2 \mbb{I}_{Z_1\geq 0}]
&=(1-\beta)\left[\mu_2^2+\sigma^2_{2}\right]+p_0\left[\mu_1\frac{\sigma_{12}^2}{\sigma_{1}^2}+2\sigma_{12}\mu_2\right] \;\\
&=(1-\beta)\left[\mu_2^2+\sigma_{2}^2\right]+p_0\sigma_{12}\left[-\Phi^{-1}(\beta)\frac{\sigma_{12}}{\sigma_{1}} +2\mu_2\right]\;\\
&:=g(\mu_1,\mu_2,\sigma_1,\sigma_2,\sigma_{12})
\end{align*}

Similarly,
\begin{eqnarray*}
\Exp[Z_2^2|Z_1=0]&=&
(\mu_2-\frac{\sigma_{12}}{\sigma_{1}^2}\mu_1)^2+\sigma^2_{2}-\frac{\sigma_{12}^2}{\sigma_{1}^2}
=(\mu_2+\frac{\sigma_{12}}{\sigma_{1}}\Phi^{-1}(\beta))^2+\sigma^2_{2}-\frac{\sigma_{12}^2}{\sigma_{1}^2}
:=h(\mu_1,\mu_2,\sigma_1,\sigma_2,\sigma_{12})
\end{eqnarray*}
\end{proof}

\bibliographystyle{cje}
\bibliography{references_cvar}

\end{document}